\documentclass[lettersize,journal]{IEEEtran}
\ifCLASSOPTIONcompsoc
\usepackage[nocompress]{cite}
\else
\usepackage{cite}
\fi
\usepackage{array}
\usepackage[caption=false,font=normalsize,labelfont=sf,textfont=sf]{subfig}
\usepackage{textcomp}
\usepackage{stfloats}
\usepackage{url}
\usepackage{verbatim}

\usepackage{hyperref}
\usepackage{graphicx}
\usepackage{textcomp}
\usepackage{amsmath,amssymb,bm,bbm,mathrsfs,amscd}
\usepackage{calc}
\usepackage{color}
\usepackage{xcolor}
\usepackage{dsfont}
\usepackage{graphicx}
\usepackage{epstopdf}
\usepackage{epsfig}
\usepackage{tikz}
\usepackage{amsmath}
\usepackage{amsfonts}
\usepackage{amssymb}
\usepackage{amsthm}
\usepackage{rotating}
\usepackage{mathtools}
\usepackage{color}
\usepackage{subfig}
\usepackage{enumerate,pgfplots}
\usepackage{algorithm}
\usepackage{algpseudocode}
\newtheorem{theorem}{Theorem}

\newtheorem{proposition}{Proposition}

\newtheorem{assumption}{Assumption}
\newcommand{\ba}{\begin{array}}
\newcommand{\ea}{\end{array}}

\newcommand{\be}{\begin{equation}}
\newcommand{\ee}{\end{equation}}

\newcommand{\ds}{\displaystyle}

\newcommand{\mc}{\mathcal}

\def\1{\mathds{1}}

\newcommand{\R}{\mathbb{R}}

\newcommand{\V}{\mathcal{V}}

\newcommand{\G}{\mathcal{G}}

\def\R{\mathbb{R}}

\def\0{\boldsymbol{0}}

\tikzstyle{v_c}=[circle, draw,inner sep=2pt, minimum width=12pt, color=blue]
\tikzstyle{v_a}=[circle, draw,inner sep=2pt, minimum width=12pt, color=red]
\tikzstyle{edge} = [draw,thick,-,font=\small ]
\tikzstyle{label} = [draw,fill=black,font=\normalsize]

\def\G{{\mathcal G}}

\def\BibTeX{{\rm B\kern-.05em{\sc i\kern-.025em b}\kern-.08em
	T\kern-.1667em\lower.7ex\hbox{E}\kern-.125emX}}
\usepackage{balance}

\begin{document}
\title{Modeling and Control of Sustainable Transitions through Opinion-Behavior Coupling in Heterogeneous Networks}
\author{Martina~Alutto, \IEEEmembership{Member, IEEE}, Sofia Bellotti, Fabrizio Dabbene, \IEEEmembership{Senior Member, IEEE}, Chiara Ravazzi, \IEEEmembership{Senior Member, IEEE}
\thanks{Martina~Alutto is Postdoctoral Researcher at KTH Royal Institute of
Technology, Stockholm (e-mail: alutto@kth.se). Fabrizio Dabbene (Director of research) and Chiara Ravazzi (Senior researcher) are with the National Research Council (CNR-IEIIT), c/o Politecnico di Torino, 10129 Torino, Italy (e-mail: {\{chiara.ravazzi;\,fabrizio.dabbene;\}@cnr.it}). 
Sofia Bellotti is a student at Politecnico di Torino and intern at CNR-IEIIT, Italy.
	}
\thanks{This work has been supported by the European Union -- Next Generation EU, Mission 4, Component 1, under the PRIN project {\em{TECHIE: A control and network-based approach for fostering the adoption of new technologies in the ecological transition}}, Cod. 2022KPHA24, CUP Master: D53D23001320006, CUP: B53D23002760006.}%
}


\maketitle
\begin{abstract}	
Understanding how sustainable behaviors spread within heterogeneous societies requires the integration of behavioral data, social influence mechanisms, and structured approaches to control. In this paper, we propose a data-driven computational framework for coupled opinion-adoption dynamics in social systems. Each node in the multilayer network represents a community characterized by a specific age group and mobility level, derived from large-scale survey data on the predisposition to adopt electric vehicles in Northern Europe. The proposed model captures three mechanisms: behavioral contagion through social and informational diffusion, abandonment driven by dissatisfaction, and feedback between opinions and adoption levels through social influence. Analyzing the equilibrium points of the coupled system allows us to derive the conditions that enable large-scale adoption. We empirically calibrate the model using data to construct synthetic populations and social similarity networks, which we use to explore targeted interventions that promote sustainable transitions. Specifically, the analysis focuses on two types of control strategies: opinion-based policies, which act on the social network layer, and policies that aim to improve experience and reduce dissatisfaction. Simulation results show that the latter ensure more stable and long-term adoption, offering concrete insights for designing effective interventions in sociotechnical transitions toward sustainability.
\end{abstract}


\graphicspath{{Plots}}
\section{Introduction} 
Global emissions remain extremely high despite a series of policy initiatives aimed at promoting sustainable mobility and growing awareness of the climate crisis in modern societies. According to a growing body of research, large-scale behavioral transitions, such as the adoption of electric vehicles (EVs), are hampered not only by economic or infrastructural barriers but also by social beliefs, peer effects, and feedback mechanisms that shape individual decision-making \cite{Coffman2017}

The use of data-informed modeling to support sustainable transitions has been the subject of more recent research in computational social systems. For example, in \cite{BRESCHI2023103651} the authors propose a data-driven, human-centered framework that helps policymakers in assessing how user attitudes, social imitation, and EV technology interact to affect the adoption of EVs naturally. Similarly, \cite{VILLA2024106106} highlight the importance of using large-scale survey data to inform mathematical models and assess how well policy interventions contribute to accelerating the adoption of green technology.

A powerful set of methodologies to analyze and exploit the interactions between personal attitudes, social norms, and external interventions is offered by opinion dynamics and control theory \cite{nylund2022enabling}. The DeGroot’s consensus process \cite{DeGroot1974} and the Friedkin-Johnsen (FJ) model \cite{Friedkin1990} are two examples of classical models of opinion dynamics that have been expanded to account for persistent biases and to design informed interventions. For instance, in \cite{piccinin2025innovation} the authors propose an extension of the FJ model to examine the long-term effects of fostering policies and to design nudging strategies that balance widespread adoption with financial restrictions. 

Although opinion dynamics models show how attitudes change in response to social influences, complementary diffusion frameworks are needed to understand how these attitudes result in tangible behavioral changes. The spread of sustainable behaviors often resembles contagion processes, where adoption is triggered and reinforced through interpersonal interactions and network effects. Numerous works that adapt innovation diffusion and epidemic-inspired models to the social domain have been inspired by this analogy \cite{daley1964epidemics, bikhchandani1992theory, Bass1969, Rogers2010}. 
Compartmental formulations such as SIS and SIR, originally developed for epidemiology \cite{Kermack.McKendrick:1927}, have been widely applied to model behavioral contagion \cite{goffman1964generalization, goffman1966mathematical, Pastor-Satorras2015}. Susceptible-Infected-Vigilant (SIV) \cite{Xu2024} and Susceptible-Infected-Recovered-Susceptible (SIRS) \cite{LI20141042, ZHANG2021126524} extensions can also capture temporary adoption and abandonment due to discontent, cost sensitivity or changing norms. Recent studies have extended these models to account for coupled epidemic-opinion dynamics on multi-layer networks, highlighting how information dissemination or opinion evolution and behavioral adoption can influence each other and affect propagation patterns \cite{granell2013dynamical, wang2019impact, lin2021discrete, She2022, pang2025dynamics, an2024coupled}.
Beyond epidemic-inspired approaches, models from economics and sociology, such as the Bass model \cite{Bass1969} and threshold-based frameworks \cite{Granovetter1978}, formalize adoption as the combined effect of individual decisions and peer influence. Despite their different origins, these models share a common idea: local interactions and spontaneous transitions on networks can drive global dynamics. More sophisticated mechanisms, including game-theoretical models, further capture the strategic nature of behavioral change \cite{montanari2010spread, young2011dynamics, morris2000contagion}.

Building on this line of research, recent works have emphasized the importance of designing intervention strategies that not only capture the coupling between adoption and opinion dynamics, but also take into account their evolution over time and their long-term sustainability. In this direction, \cite{alutto2025mpc} propose a predictive control framework for a multilayer adoption-opinion model, where adoption evolves through social contagion and perceived benefits, while opinions are shaped by peer influence and feedback from adoption levels. The model also includes dissatisfaction and abandonment mechanisms that can slow down diffusion. Within this setting, the authors propose a Model Predictive Control approach that dynamically adapts interventions to the predicted trajectory of the system. Different types of control (shaping opinions, adoption enhancement, and dissatisfaction reduction) are compared and the results show that predictive strategies consistently outperform static ones, achieving higher adoption at comparable or lower cost. This highlights the potential of predictive and adaptive control methods as effective and scalable tools for sustaining innovation diffusion and promoting long-term behavioral change.

\subsection{Our contribution}
In this work, we build upon the adoption-opinion framework, originally discussed in \cite{alutto2025mpc} (see also the extended version in \cite{alutto2025predictivecontrolstrategiessustaining}), and propose a data-driven extension specifically designed to account for heterogeneous societies. Our main innovation lies in explicitly accounting for heterogeneity through a distributional structure: rather than assuming homogeneous agents, we construct a synthetic population where each individual belongs to a specific mobility-age subgroup. This modeling choice is motivated not only by theoretical considerations, but also by empirical evidence: mobility habits and age have been identified as key adoption factors in analyses based on real-world population data. The mobility level is particularly relevant because highly mobile groups are more visible within social networks and can exert stronger influence on adoption dynamics, whereas less mobile groups tend to occupy more peripheral positions. Age, in turn, captures systematic differences in openness to innovation, risk perception, and responsiveness to peer influence. By incorporating these two factors jointly, our model provides a richer and more realistic representation of how adoption spreads in heterogeneous societies.  

Within this framework, we first analyze the equilibrium points with and without widespread adoption in the network and their stability. Then, we integrate survey-based data to initialize the population distribution, build the similarity social network, and infer key behavioral parameters. This yields a data-driven computational model whose simulations capture both structural heterogeneity and empirically based behavioral tendencies. Building on this foundation, we design and test different intervention policies aimed at accelerating EV adoption. We focus on two levers: (i) opinion-based interventions, which shape public attitudes and perceptions of EVs, and (ii) dissatisfaction-based interventions, which instead target adopters’ experience by reducing the risk of disaffection. These simulation results reveal the trade-offs between these approaches. While opinion shaping interventions can accelerate adoption, they also risk amplifying dissatisfaction if user experience is not improved. In contrast, dissatisfaction-focused policies lead to a more sustained level of adoption over time. These results demonstrate the potential of computational social modeling for analyzing heterogeneous behavioral systems and informing data-driven policy design.

\subsection{Outline}
The rest of the paper is organized as follows. Section~\ref{sec:model} introduces the adoption-opinion model and provides stability analysis of the equilibrium points. Section~\ref{sec:data} presents the survey-based data analysis, while Section~\ref{sec:comparison} compares different targeted control strategies. Finally, Section~\ref{sec:conclusion} concludes the paper and outlines directions for future research.

\subsection{Notation}
We denote by $\R$ and $\R_{+}$ the sets of real and nonnegative real numbers. 
The all-1 vector and the all-0 vector are denoted by $\1$ and $\0$ respectively. The identity matrix and the all-0 matrix are denoted by $I$ and $\mathbb{O}$, respectively.
The transpose of a matrix $A$ is denoted by $A^T$. 
For $x$ in $\R^n$, let $||x||_1=\sum_i|x_i|$ and $||x||_{\infty}=\max_i|x_i|$ be its $l_1$- and $l_{\infty}$- norms. 
For an irreducible matrix $A$ in $\R_+^{n\times n}$, we let $\rho(A)$ denote the spectral radius of $A$. 
Inequalities between two matrices $A$ and $B$ in $\mathbb{R}^{n \times m}$ are understood to hold entry-wise, i.e., $A \leq B$ means $A_{ij} \leq B_{ij}$ for all $i$ and $j$, while $A < B$ means $A_{ij} < B_{ij}$ for all $i$ and $j$. 

\section{Model description and analysis}\label{sec:model}
In this section, we present the adoption-opinion model, originally introduced in \cite{alutto2025mpc}, with a modified formulation at the physical layer. 

The model describes the diffusion of sustainable behaviors, such as electric vehicle adoption, within a closed population of total size $N$, assuming no inflow or outflow of individuals over time.  The population is partitioned into $n$ distinct communities, each identified by a pair of indices $(m,p)$, where $m$ denotes a mobility category and $p$ a socio-demographic characteristic. Each community is then associated with a joint distribution $f_{m,p}$, representing the proportion of individuals belonging to the corresponding group. The mobility index captures the extent to which individuals interact and are visible within the social network, thereby influencing the impact of their adoption decisions. At the same time, socio-demographic attributes, such as age or income, play a crucial role, as they reflect differences in openness to innovation, risk perception, but also responsiveness to peer influence. By jointly considering these two dimensions, the model incorporates heterogeneity in both social visibility and behavioral attitudes, yielding a more realistic representation of adoption dynamics.

Within each community $(m,p)$, individuals are further classified into three dynamic compartments. The variable $s_{m,p}(t)$ denotes the fraction of susceptible individuals at time $t \in \mathbb{Z}_+$, who have not yet adopted a given behavior or technology but may do so in the future, belonging to the community $(m,p)$. The fraction $a_{m,p}(t)$ represents those who are currently adopting the behavior or using the service, while $d_{m,p}(t)$ represents dissatisfied individuals who have previously adopted but abandoned the behavior due to dissatisfaction, or who perceive adoption as unfavorable and are unlikely to adopt in the next future. The three fractions together account for the entire community, so that their sum reflects the distribution $f_{m,p}$. This compartmental structure allows the model to capture both the spread of adoption and the dynamics of abandonment within heterogeneous communities. 
In addition, each community is associated with a dynamic opinion variable $x_{m,p}(t)\in[0,1]$, representing its average attitude toward the innovation or service, as the electric vehicle adoption, at time $t \in \mathbb{Z}_+$. 


\subsection{Adoption dynamics}
\begin{figure}
	\centering
	\includegraphics[scale=0.5]{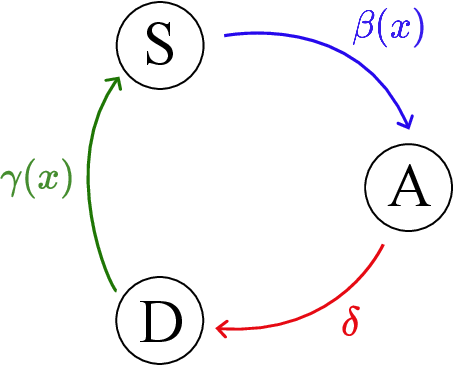}
	\caption{Schematic representation of the three-state adoption model. Individuals can transition between compartments through the indicated flows, describing the dynamics of adoption, abandonment, and potential re-adoption.}
	\label{fig:sirs}
\end{figure}

The adoption dynamics account for both local opinions and the aggregate influence of adopters across all communities. This influence is modeled not only through the prevalence of adopters but also by the mobility levels of both adopters and susceptibles. On the one hand, adopters from highly mobile groups tend to be more visible in daily interactions and therefore contribute disproportionately to the perceived prevalence of adoption. On the other hand, susceptible communities with higher mobility levels are more exposed to such interactions, and thus their propensity to adopt is amplified accordingly. Such a mechanism reflects empirical evidence that socially exposed or mobile individuals often act as multipliers or leaders in the diffusion of innovations, while at the same time being more likely to encounter persuasive adopters in their daily routines \cite{Rogers2010}. 
The evolution of adoption dynamics across interacting communities is described by the following discrete-time system:
\begin{align}\label{eq:adoption-model}
s_{m,p}(t+1) &= s_{m,p}(t) +\gamma(x_{m,p}(t)) d_{m,p}(t) +\nonumber \\
& -\beta(x_{m,p}(t)) m s_{m,p}(t)\ds \sum_{m',\, p'} \frac{m'\, a_{m',p'}(t) }{ \sum_{m'',\, p''} m'' f_{m'', p''}}  \nonumber\\[1pt]
a_{m,p}(t+1) &= a_{m,p}(t) - \delta_{m,p}a_{m,p}(t) +\\
& + \beta(x_{m,p}(t))m  s_{m,p}(t) \ds \sum_{m',\, p'} \frac{m'\, a_{m',p'}(t) }{ \sum_{m'',\, p''} m'' f_{m'', p''}} \nonumber\\[1pt]
d_{m,p}(t+1) &= d_{m,p}(t) - \gamma(x_{m,p}(t))d_{m,p}(t) + \delta_{m,p} a_{m,p}(t)\,,\nonumber
\end{align}
for each community $(m,p)$ in $\mathcal{V} = \{(m,p) \mid m \in \mathcal{M},\, p \in \mathcal{P}\}$, where $\mc M$ and $\mc P$ denote the index sets for the mobility and socio-demographic characteristic, respectively. The \emph{adoption rate} $\beta(x_{m,p}(t))$ and the \emph{reversion rate} $\gamma(x_{m,p}(t))$ are both modeled as increasing functions of the local average opinion $x_{m,p}(t) \in [0,1]$. A more favorable opinion leads to higher adoption and reversion probabilities. For simplicity, in this work we will assume $\beta(x_{m,p}) = \beta x_{m,p}$, $\gamma(x_{m,p}) = \gamma x_{m,p}$ for some $\beta \in [0,1]$, $\gamma \in [0,1)$. 
Finally, the parameter $\delta_{m,p} \in [0,1]$ represents the \emph{dismissal rate}, that is the probability that an adopter in community $(m,p)$ becomes dissatisfied at the next time step.

Figure \ref{fig:sirs} shows the three-state adoption model where the transitions between different compartments are shown with arrows. Note that the adoption opinion model in \eqref{eq:adoption-model} differs from the model in \cite{alutto2025mpc} both in the absence of a transition from the susceptible state to the dissatisfied state and in the different modeling of adoption dynamics. Here, all adopters in the various communities influence the susceptible individuals in proportion to their mobility index, without specifying a physical layer formed by neighborhoods.

\subsection{Opinion dynamics}
Social interactions are mediated by a network structure that reflects similarity in attitudes rather than geographical proximity.
To formalize this, we describe the system as a directed graph ${\mc G}= (\mc V, {\mc E}, W)$, where ${\mc E} \subseteq \mc V \times \mc V$ denotes the set of directed links and $W \in \R_+^{n \times n}$ is a nonnegative weight matrix, referred to as the social interaction matrix. 
Building on this, we extend the Friedkin-Johnsen framework of opinion dynamics~\cite{Friedkin1990}, by incorporating a feedback channel that links opinions to actual adoption levels:
\begin{align}\label{eq:opinion-model} 
x_{m,p}(t+1)\!\! =&\alpha_{m,p} x_{m,p}(0) \!+\! \lambda_{m,p} \mspace{-6mu} \ds \sum_{m',\, p'}\mspace{-6mu} {W}_{(m,p),(m',p')}x_{m',p'}(t) +\nonumber\\[2ex]
&\, + \xi_{m,p}  \mspace{-6mu} \ds \sum_{m',\, p'} \frac{m' a_{m',p'}(t)}{\ds \sum_{m'',\, p''} m'' f_{m'', p''}} \,.
\end{align}
The dynamics of opinions thus emerge from the interplay of three mechanisms: the persistence of individual predispositions, peer influence within social networks, and feedback based on the extent of adoption, again weighted by the adopters' degree of mobility.
The weights $\alpha_{m,p}$, $\lambda_{m,p}$, and $\xi_{m,p}$ are nonnegative and satisfy $\alpha_{m,p} + \lambda_{m,p} + \xi_{m,p} = 1$, determining the relative weight of the three drivers:
\begin{itemize}
\item the term $\alpha_{m,p} x_{m,p}(0)$ reflects the \textit{inertial tendency} of a community to preserve its initial belief;
\item the second term captures \textit{social influence}, where $\lambda_{m,p}$ scales the impact of neighboring communities in the similarity network;
\item the third term introduces a \textit{behavioral feedback mechanism}, whereby communities adjust their opinions in response to adoption levels, with higher mobility groups exerting stronger influence.  
\end{itemize}
The parameters $\lambda_{m,p}$ and $\xi_{m,p}$ modulate how responsive each community is to social and behavioral influences. A larger $\lambda_{m,p}$ indicates a stronger alignment with peers in the network, while a higher $\xi_{m,p}$ makes opinions more sensitive to the overall adoption levels in the population. This formulation closes the feedback loop between adoption and opinion: local attitudes shape adoption decisions \eqref{eq:adoption-model} and the resulting spread of adoption feeds back into opinions, either reinforcing or slowing down further diffusion.

\subsection{Coupled adoption-opinion model}
The adoption and opinion dynamics in \eqref{eq:adoption-model}-\eqref{eq:opinion-model} are intrinsically coupled.  
Let $s$, $a$, $d$, and $x$ denote the vectors containing, respectively, the fractions of susceptibles, adopters, dissatisfied individuals, and average opinions in each of the $n$ communities. For notational convenience, we introduce a single index $i \in {1,\dots,n}$ that uniquely identifies each community, corresponding to a specific pair $(m,p)$ defined by a mobility level and a socio-demographic group.

\begin{assumption}\label{ass:ass1}
	If, for any node $i\in\V$, there exists a path in ${\G}$ from $i$ to $j$ with $\lambda_j<1$ and $x_j(0)>0$.
\end{assumption}

We first establish a well-posedness result, whose proof is omitted for brevity.
\begin{proposition}\label{prop:invariant}
Consider the adoption-opinion model \eqref{eq:adoption-model}-\eqref{eq:opinion-model}. Then, if $s(0), a(0), d(0)$ in $[0,1]^{\V}$ and $s(0)+ a(0)+ d(0) = \1$, then $s(t), a(t),d(t)$ in $[0,1]^{\V}$ and $s(t)+a(t)+d(t)=\1$ for all $t\geq0$. Moreover, if $x(0)$ in $[0,1]^{\V}$, then $x(t)$ in $[0,1]^{\V}$ for all $t \geq 0$.
\end{proposition}
As a direct consequence of Proposition~\ref{prop:invariant}, the adoption state can be fully described using just two variables (e.g., adopters and dissatisfied), together with the opinion vector.  
Let $\mathrm{m} \in \mathbb{R}^{\mc V}$ be the vector of mobility indices, and define the normalization constant
$$c = \frac{1}{\ds \sum_{m'',\, p''} m'' f_{m'', p''}}\,,$$
which rescales the aggregate influence of adopters, proportionally to their mobility level.  
We can then rewrite the model in compact vectorial form,
\begin{align}\label{eq:vector_model} 
a(t\!+\!1)\! \!&=\! 
[I\!+\! c \beta \mathrm{diag}( \!x(t)\! )M \mathrm{diag}(\1\! -\! a(t)\!-\! d(t)\!) \1 \! \1^T\!\! M \!-\! \Delta] a(t) \nonumber \\
d(t\!+\!1)\!\!&=\! d(t) - \gamma \mathrm{diag}(x(t)) d(t)  + \Delta a(t) \\
x(t\!+\!1) \!\!&=\! (I-\Lambda-\Xi) x(0) + \Lambda W x(t) + c \Xi \1 \mathrm{m}^T a(t),\nonumber
\end{align}
where $\Delta =  \mathrm{diag}(\delta)$, $M= \mathrm{diag}(\mathrm{m})$, 
$\Lambda=\mathrm{diag}(\lambda)$, $\Xi=\mathrm{diag}(\xi)$.

\subsection{Equilibrium points and stability analysis}
In this subsection, we investigate the equilibrium points of the system and we analysis their stability, following the analysis done in \cite{alutto2025mpc}. To this end, it is useful to introduce a threshold quantity that summarizes the potential of an innovation to spread, in close analogy with epidemic models.

Indeed, in epidemic modeling the \emph{basic reproduction number} plays this role. It quantifies the expected number of secondary cases generated by a single infectious individual in an fully susceptible population. In the context of behavioral diffusion, a similar threshold indicates whether a sustainable innovation, such as the adoption of electric vehicles, can spread widely across the population or remain confined to a small group of early adopters.
For the dynamics in \eqref{eq:vector_model}, we introduce the \emph{opinion-dependent reproduction number}:
\begin{equation*}
R_0^A(t) = \rho\left(I - \Delta + c \beta\, \mathrm{diag}(x(t)) \mathrm{m}\, \mathrm{m}^T \right)\,,
\end{equation*}
which evolves over time, since it explicitly depends on the opinion vector $x(t)$. This formulation captures how social beliefs dynamically influence the potential spread of adoption.

To analyze stability, we consider two extreme scenarios regarding public opinion.
First, assume that all communities strongly support adoption. Let $\bar{x} \in [0,1]^{\mathcal{V}}$ denote the vector of upper bounds of the opinion states $x(t)$ for all $t \ge 0$ which can be derived by the dynamics in \eqref{eq:opinion-model} as 
\be \label{eq:barx} x(t) \leq \bar x := (I-\Lambda-\Xi) x(0) + \Lambda \1 + \Xi \1,\ee 
for all $t \geq 0$. In this optimistic scenario, the system attains its maximal potential for behavioral diffusion, and we define the \emph{maximal reproduction number} as
\begin{equation}\label{eq:r0-max}
R_{0, \max}^A = \rho\left(I - \Delta + c\beta \, \mathrm{diag}(\bar{x})\mathrm{m} \, \mathrm{m}^\top \right),
\end{equation}
which provides an upper bound on the ability of the innovation to spread when attitudes are most favorable.
Conversely, consider a population in which most communities are predominantly skeptical or resistant. Let $\underline{x} \in [0,1]^{\mathcal{V}}$ denote the vector of lower bounds of the opinion states as follows
\be \label{eq:underx} x(t) \geq \underline x := (I-\Lambda-\Xi) x(0), \ee 
for all $t \geq 0$. In this pessimistic scenario, we define the \emph{minimal reproduction number}:
\begin{equation}\label{eq:r0_min}
R_{0, \min}^A = \rho\left(I - \Delta + c \beta \, \mathrm{diag}(\underline{x}) \mathrm{m}\, \mathrm{m}^\top  \right).
\end{equation}
This provides a lower bound on the system’s capacity to initiate and sustain adoption under unfavorable collective attitudes.
Thus, the two quantities $R_{0,\max}^A$ and $R_{0,\min}^A$ thus provide threshold conditions that will guide the subsequent stability analysis.  

We now introduce some auxiliary quantities used in the stability analysis. First define
\begin{equation}\label{eq:nu}
\nu := \min_{i\in\mc V} \left\{ \delta_i \right\}, \qquad
\eta := 1 - \gamma \inf_{t} \| x(t)\|_\infty ,
\end{equation}
and for a given equilibrium $(a^\dag, d^\dag, x^\dag)$, define
\begin{equation}\label{eq:varphi}
\varphi \!:= \mspace{-15mu} \max_{a,d,x \in [0,1]^{\mathcal{V}}}\! \left\| \! I \! -\!\! \Delta \!\! -\!\! \mathcal{B}^\dag \!\!+ \!\! c \beta M \mathrm{diag}(x) \mathrm{diag}(\1 \!-\! a\! -\! d)\! \1 \! \1^T \!  M \right\|_{\infty}\!\!,
\end{equation}
where $\mathcal{B}^\dag = c \beta \, \mathrm{diag}\left( M\mathrm{diag}(x(t)) \1 \1^\top M a^\dag \right)$. Finally, we introduce the matrix
\begin{equation}\label{eq:G} G:= \begin{bmatrix}
		\varphi & \sqrt{\sup_t \| \mc B^\dag\|_\infty \nu}\\
		\sqrt{\sup_t \| \mc B^\dag \|_\infty \nu} & \eta
\end{bmatrix}.\end{equation}

The following result, whose proof can be found in Appendix, presents a stability analysis of the system’s equilibrium points. In particular, we focus on two configurations: the \emph{adoption-free equilibrium}, in which no individuals have adopted the innovation, and the \emph{adoption-diffused equilibrium}, where a nonzero fraction of the population has adopted. This analysis characterizes the conditions under which adoption can emerge, persist, or die out in the network.
\begin{theorem}\label{theo:stability}
Consider the adoption-opinion model \eqref{eq:vector_model}. Then the following facts hold.
\begin{enumerate}
	\item[(i)] There exists an adoption-free equilibrium of the form $(\0, \0, x^*)$, where $x^* = (I - \Lambda W)^{-1} (I - \Lambda - \Xi) x(0)$.
	Moreover:
	\begin{itemize}
		\item[(a)] if $R_{0}^A (x^*)< 1$, the adoption-free equilibrium is locally stable,
		\item[(b)] if $R_{0}^A  (x^*)> 1$, the adoption-free equilibrium is unstable.
		\item[(c)] if $R_{0,\max}^A <1$, the adoption-free equilibrium is globally asymptotically stable.
	\end{itemize}
	\item[(ii)] If $R_{0}^A  (x^*)> 1$, then there exists at least one adoption-diffused equilibrium $(a^\dag, d^\dag, x^\dag)$ with $a^\dag > \0$. 
	If $\rho(G) <1$,
	then the equilibrium is locally asymptotically stable for any non-zero initial condition.
\end{enumerate}
\end{theorem}

The analysis above clarifies the conditions governing the existence and stability of equilibria in the coupled adoption-opinion system. Building on these theoretical results, we now turn to data-driven analysis, where model parameters are estimated and the proposed framework is confronted with survey evidence to assess its empirical validity.

\section{Data Analysis}\label{sec:data}
This section presents an analysis of data collected via a transnational survey that investigated knowledge and attitudes toward electric vehicles among Northern European consumers. 

An initial version of the survey was carried out in Denmark in 2018. The survey aimed to explore the public's knowledge of and attitudes toward electric vehicles, as well as to identify common beliefs and misconceptions. The same survey was later replicated in the Netherlands (December 2019), and again in Denmark, as well as in Germany, Hungary, and Norway (May 2020). The questionnaire consisted of 51 multiple-choice questions, divided into three sections: 9 questions addressed socio-demographic characteristics, 19 addressed mobility habits, and 23 addressed common myths about EVs. The survey was administered online to approximately 1500 respondents per country, aged between 18 and 70. The complete dataset is publicly available in \cite{mendeleyEVdata}. The overarching goal of this data collection effort was to compare perceptions and knowledge about EVs across countries with different levels of market maturity.

\subsection{Clustering socio-demographic data}
To identify the most representative socio-demographic variables, we use an unsupervised learning approach that does not rely on a predefined target variable. Specifically, we apply clustering techniques using the well-known \textit{K-Means} algorithm to detect natural groupings in data. K-Means is an iterative algorithm that partitions a dataset into $k$ clusters by minimizing the distance between each point and the centroid of its assigned cluster.
The clustering procedure was implemented using the \texttt{KMeans} method from the scikit-learn library \cite{scikitKMeans}.

A key step in clustering is selecting the number of clusters $k$. To guide this selection, we employ two standard techniques: the \textit{elbow method} and\textit{ silhouette coefficient analysis}. The elbow method evaluates the Within-Cluster Sum of Squares (WCSS) as a function of $k$, identifying the point at which adding more clusters yields diminishing reductions in WCSS. As shown in Figure~\ref{gomito} (left), the largest decrease occurs up to $k=5$, although the inflection point is not sharply defined. To complement this assessment, we compute the silhouette coefficient, which measures how well each data point fits within its cluster relative to other clusters. Higher silhouette scores indicate more compact and well-separated clusters. As illustrated in Figure~\ref{gomito} (right), the silhouette score reaches its maximum at $k=5$, which supports this as the most appropriate clustering configuration.

\begin{figure}
	\includegraphics[width=4.2cm,height=3.2cm]{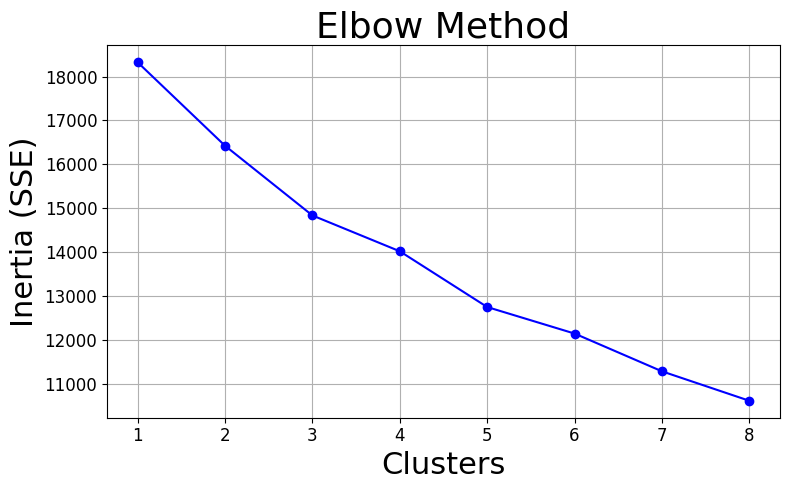}
	\includegraphics[width=4.2cm,height=3.2cm]{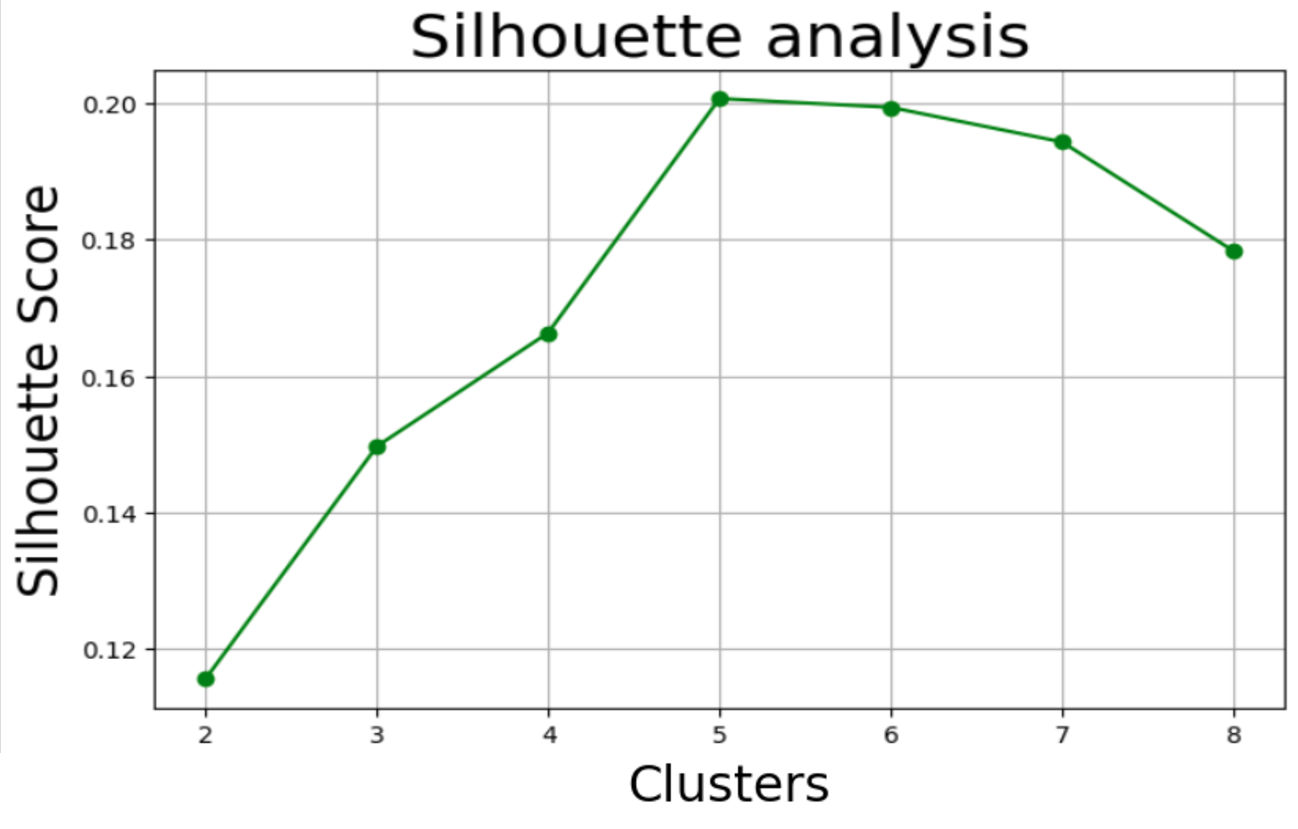}
	\caption{Determination of the optimal number of clusters for socio-demographic data. Left: Elbow method based on Within-Cluster Sum of Squares (WCSS). Right: Silhouette coefficient analysis. Both methods indicate $k=5$ as the most appropriate clustering configuration.}
	\label{gomito}
\end{figure}

Based on these results, we apply \textit{K-Means} with $k=5$ to the socio-demographic variables. Cluster labels are then assigned to each respondent, allowing us to characterize the main features of the identified groups. For each cluster, we compute the mean values of the socio-demographic variables and highlight those with the highest relative contributions. The resulting clusters can be summarized as follows:

\begin{itemize}
	\item \textit{Cluster 0}: 18-24 years old, predominantly secondary education, with left-wing or center political orientation, representing young adults likely beginning to engage with mobility choices.
	\item \textit{Cluster 1}: 45-54 years old, secondary education, left-wing and center political views, corresponding to mid-aged adults with established routines and moderate political engagement.
	\item \textit{Cluster 2}: 25-34 years old, high incidence of tertiary education, mostly left-wing political orientation, young professionals potentially more open to adopting innovative technologies.
	\item \textit{Cluster 3}: 55-70 years old, lower levels of education (secondary or below), and left-wing or center political alignment, older adults with potentially lower propensity for technological adoption.
	\item \textit{Cluster 4}: 35-44 years old, mixed education levels (secondary and tertiary), mostly left or center politically oriented, representing a mid-aged cohort with heterogeneous educational backgrounds and mobility behaviors.
\end{itemize}

Age emerges as the most discriminative factor across all clusters, shaping the distribution and heterogeneity of socio-demographic characteristics. To better visualize the separation and cohesion among the clusters, we apply Principal Component Analysis (PCA) to reduce the dimensionality, preserving most of the data's variance while enabling graphical representation. Figure~\ref{fig:enter-label} shows the resulting 3D projection of the socio-demographic clusters, highlighting the relative positioning and overlap of different population segments.

\begin{figure}
	\centering
	\includegraphics[width=0.7\linewidth]{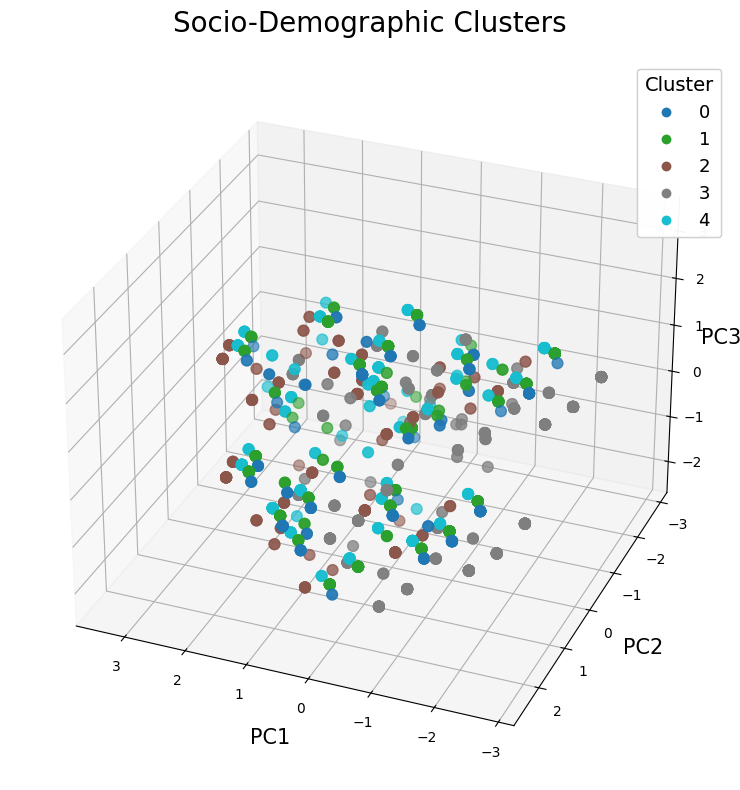}
	\caption{3D projection of socio-demographic clusters obtained via Principal Component Analysis (PCA). The five clusters identified by the K-Means algorithm are shown, highlighting the separation and cohesion of groups based on age, education, and political orientation.}
	\label{fig:enter-label}
\end{figure}

\subsection{Construction of distribution}
To create synthetic communities that capture both socio-demographic characteristics and mobility behaviors, we refine the classification of individuals further by segmenting each age group into five mobility subcategories, based on self-reported annual travel distances. This additional subdivision enables us to more accurately represent heterogeneity in mobility patterns, which is a key factor in modeling adoption processes. Individuals or groups with higher mobility are more likely to interact with a wider range of people, act as visible adopters of new technologies such as electric vehicles, and influence the perceptions and behaviors of the broader population. 

To formalize this concept, we compute the joint distribution $f_{m,p}$, which represents the fraction of individuals belonging to each mobility-age subgroup indexed by $(m,p)$. Here, $p \in \mathcal{P}$ denotes the age group, capturing generational differences in preferences, attitudes, and openness to innovation, while $m \in \mathcal{M}$ represents the mobility class, reflecting typical annual travel distance and exposure to diverse social contexts.
The age groups considered are: $\mathcal{P} = \{ \text{18--24}, \text{25--34}, \text{35--44}, \text{45--54}, \text{55--70} \}$, while the mobility range categories, derived from the survey data, are: $\mathcal{M} = \{ \text{$<$5.000}, \text{5.000-9.999}, \text{10.000-19.999}, \text{20.000-29.999}, \text{$\geq$30.000} \}$. This classification yields a total of $n = 25$ subpopulations per country.
Figure~\ref{fig:distributions} displays the empirical distributions $f_{m,p}$ for the five countries under study. In each subplot, the $x$-axis represents the age group $p$, the $y$-axis corresponds to the annual distance traveled (in km), while the vertical scale (represented by color intensity or bar height) indicates the relative size of each $(m,p)$ community within the national sample.

\begin{figure}
	\centering
	\subfloat[][]{\includegraphics[width=0.45\linewidth]{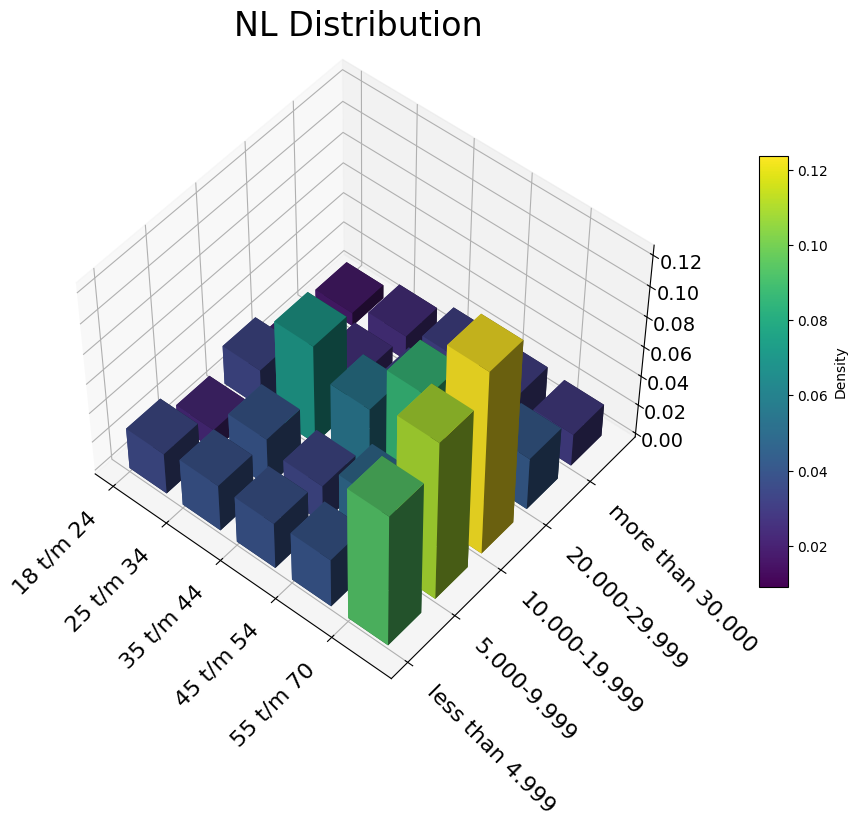}
	\label{fig:nl-year}
	}
	\subfloat[][]{\includegraphics[width=0.45\linewidth]{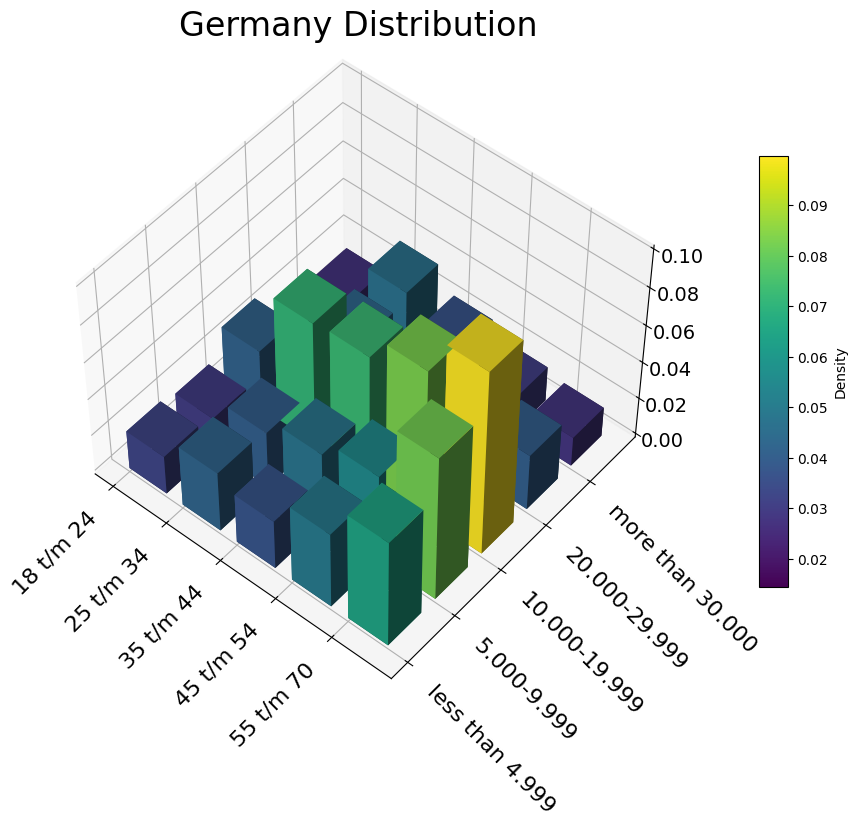}
		\label{fig:germany-year}
	}\\
	\subfloat[][]{\includegraphics[width=0.45\linewidth]{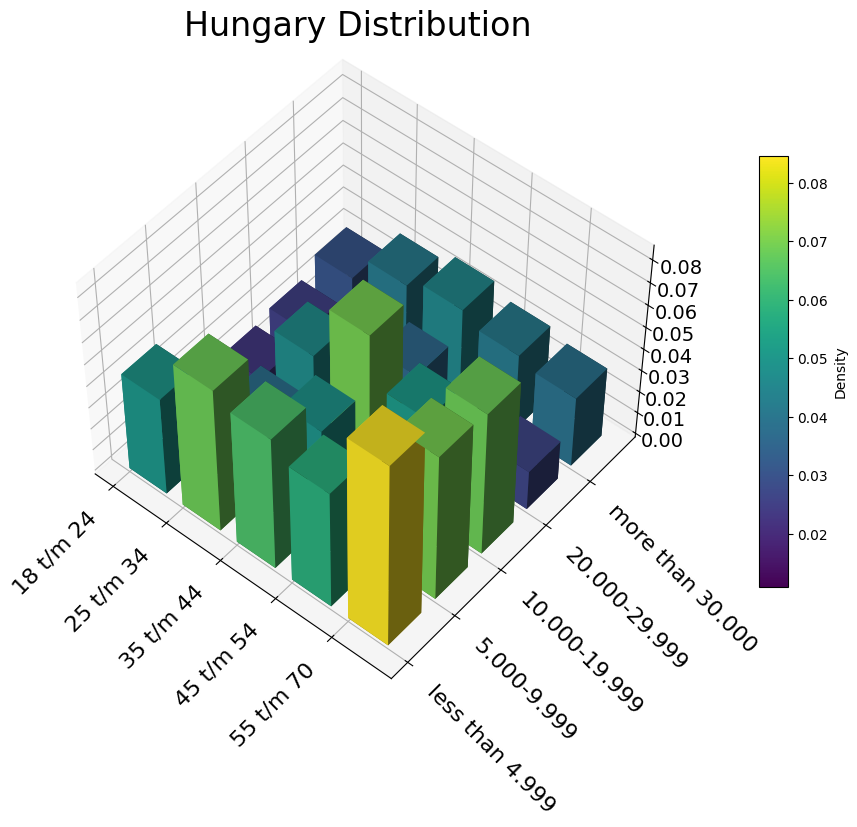}
		\label{fig:hungary-year}
	}
	\subfloat[][]{\includegraphics[width=0.45\linewidth]{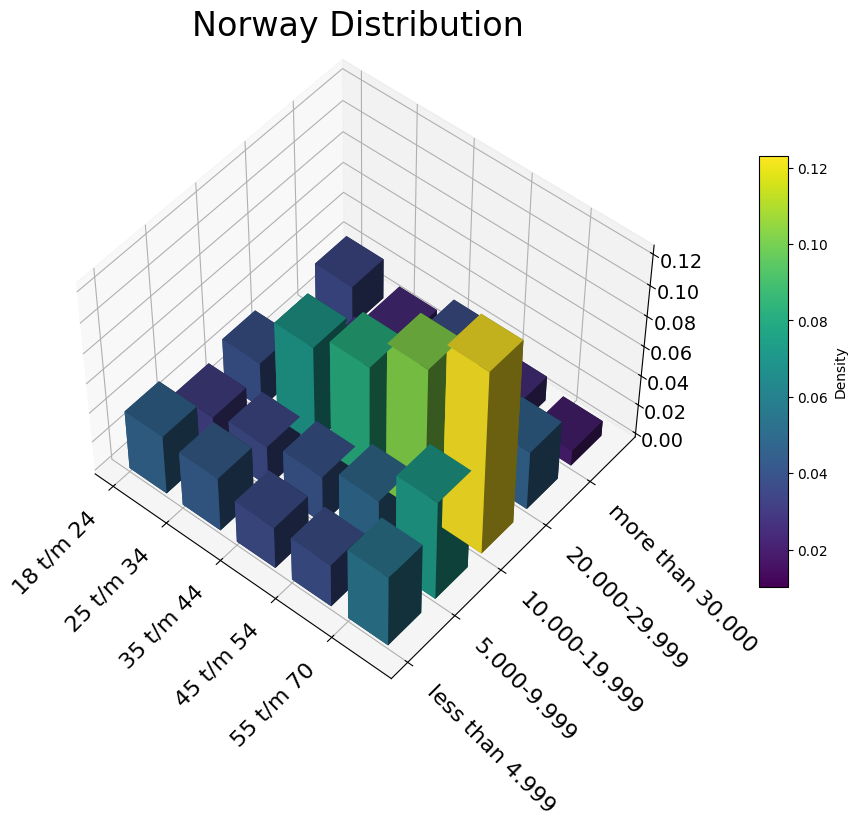}
		\label{fig:norway-year}
	}\\
	\subfloat[][]{\includegraphics[width=0.5\linewidth]{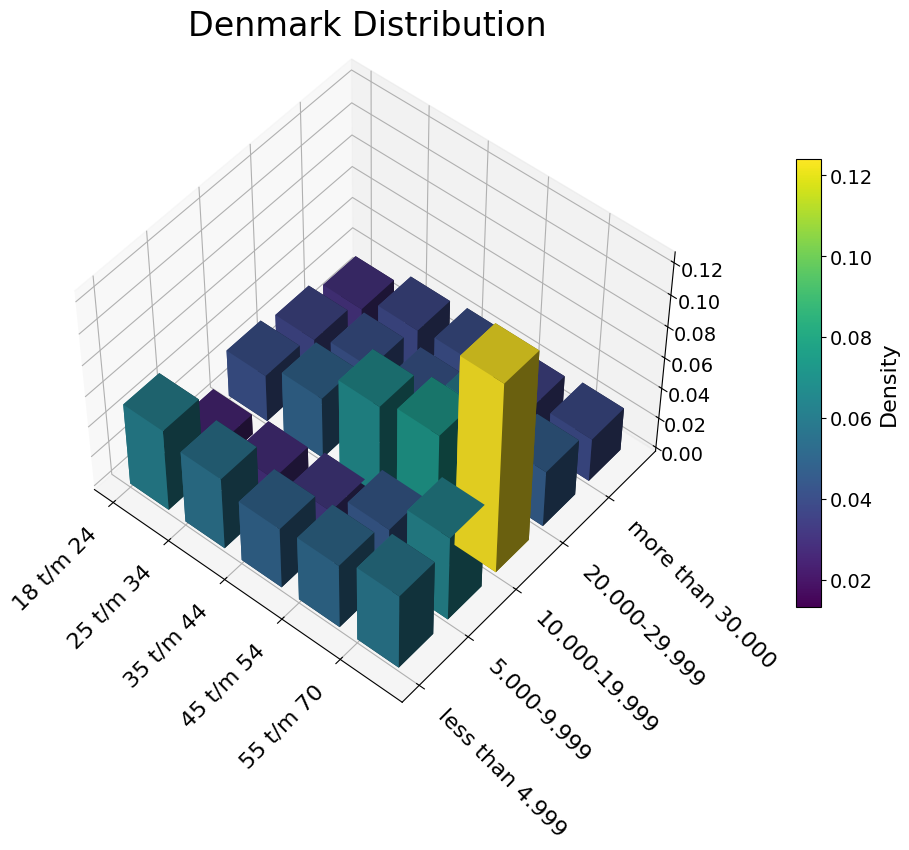}
		\label{fig:denmark-year}
	}
	\caption{Distribution of self-reported annual travel distances by age group in five countries. The $x$-axis represents age groups, the $y$-axis indicates annual distance traveled (km), and the $z$-axis shows the fraction of respondents in each group.}
	\label{fig:distributions}
\end{figure}
Significant heterogeneity is evident across countries. For instance, approximately $10\%$ of respondents in Norway and Germany are aged 45-54 and report driving between 10.000 and 19.999 km per year, while only a small fraction travel less than 5.000 km annually. By contrast, in Hungary, shorter annual distances are more common across all age groups. Denmark exhibits a different pattern: younger individuals, aged 18-34, tend to have shorter annual travel distances, whereas in most other countries, the least mobile groups are those over 55.
These differences highlight the importance of considering both age and mobility patterns when analyzing behavioral trends and adoption potential across countries.

\subsection{Social network construction}
The adoption of sustainable behaviors is influenced not only by mobility and socio-demographic characteristics, but also by the way communities influence one another.
Public opinions about electric vehicles evolve as communities continuously compare their views with those of others and with the observed adoption levels in their environment. To capture these mechanisms, we build a data-driven social network in which nodes represent socio-demographic communities and edge weights reflect the similarity of their attitudes toward EV adoption.

We start with survey data and select items that capture perceptions, beliefs, and concerns about EVs. These items cover three main dimensions: (i) common misconceptions (e.g., \textit{“EVs are boring”, “Not safe”, “Fire risk is higher”
}), (ii) technological skepticism (e.g., \textit{“Battery issues”, “Too few models”, “Limited range or speed”, “Weather affects range” 
}), and (iii) doubts about compatibility with everyday life (e.g., \textit{“Only good for short trips”, “Only suitable as second car”, “Hard to travel across borders”}). Additional evaluative questions assess what people value when buying a car, which benefits could improve their perception of EVs, and how well electric vehicles fit into daily routines.

All responses are standardized so that higher values consistently indicate a more favorable orientation toward EVs, ensuring consistency across positively and negatively worded items. For each community $(m,p)$, defined by mobility index $m$ and age group $p$, we compute an aggregate opinion profile $x_{m,p} \in [0,1]^k$, where $k$ is the number of selected items. This profile will be used to construct a data-driven social network. Individual responses are averaged at the community level and then normalized to the interval $[0,1]$, making the resulting profiles comparable across groups with different response distributions. Non-informative answers such as \textit{"Don't know"} or \textit{"Neither agree nor disagree"}, usually encoded with special or midpoint values, are recoded to a neutral response (e.g., $3$ on a $1$-$5$ scale). This approach preserves valuable information while avoiding bias from ambivalent answers, thus ensuring statistical robustness. Collecting all profiles yields a matrix of opinions, with rows representing communities and columns corresponding to standardized survey items.

To translate opinion profiles into social interactions, we construct a weighted adjacency matrix $W$ where each entry represents the similarity between two communities in opinion space. This step allows us to define a notion of \textit{social proximity} which is independent of geographic location and instead reflects how similarly different subpopulations perceive EV adoption.
Specifically, the similarity between two communities $(m,p)$ and $(m',p')$ is defined by a Gaussian (RBF) kernel applied to the Euclidean distance between their opinion profiles $x_{m,p}$ and $x_{m',p'}$. Then, each entry $ W_{(m,p), (m',p')} $ can be constructed as follows:
\be  \begin{cases}
	 \exp \left( - \frac{\| x_{m,p} - x_{m',p'} \|_2^2}{2\sigma^2} \right)\!, & \text{if } \| x_{m,p} - x_{m',p'} \|\leq 0.9 \\
	0 & \text{ otherwise}
\end{cases}\ee
where $\| x_{m,p} - x_{m',p'} \|_2$ measures the dissimilarity in opinion space, and the parameter $\sigma > 0$ controls how rapidly the similarity decays with increasing opinion distance. 
This kernel-based approach is commonly used in Graph Signal Processing (GSP) \cite{8347162, 6409473, 6494675} to build weighted graphs in which edge weights encode the similarity between nodes. In our context, it defines a \textit{social network over the opinion space}: each population is a node, and the strength of the connection between two nodes reflects the alignment of their attitudes toward EVs. 

The cutoff ensures that the kernel is applied only to pairs of communities with reasonably aligned opinions. Communities that are too far apart are considered socially disconnected, meaning they do not exert meaningful influence on each other. From a modeling perspective, this prevents the creation of an unrealistically dense network in which even very dissimilar groups would appear weakly connected. Instead, the resulting graph highlights the local structure of influence, which is more consistent with real-world social dynamics wherein interactions are stronger and more relevant among groups with similar attitudes.

The parameter $\sigma$ tunes the connectivity structure of this network: smaller values emphasize local similarities and yield sparser effective connections, while larger values lead to denser graphs that aggregate information across broader opinion differences. To calibrate $\sigma$ in a data-driven way, we use the median pairwise distance across all population pairs:
$$ \sigma = \text{median} \left( \left\{ \| x_{m,p} - x_{m',p'}\|_2 \mid (m,p)\ne (m',p') \right\} \right).
$$

This construction produces a social network that captures inter-community relationships shaped by shared or divergent beliefs about EVs, rather than by physical proximity. Figure~\ref{fig:heatmap} shows a heatmap illustrating the intensity of social interactions between communities in Germany.\medskip
\begin{figure}
	\centering
	\includegraphics[width=0.55\linewidth]{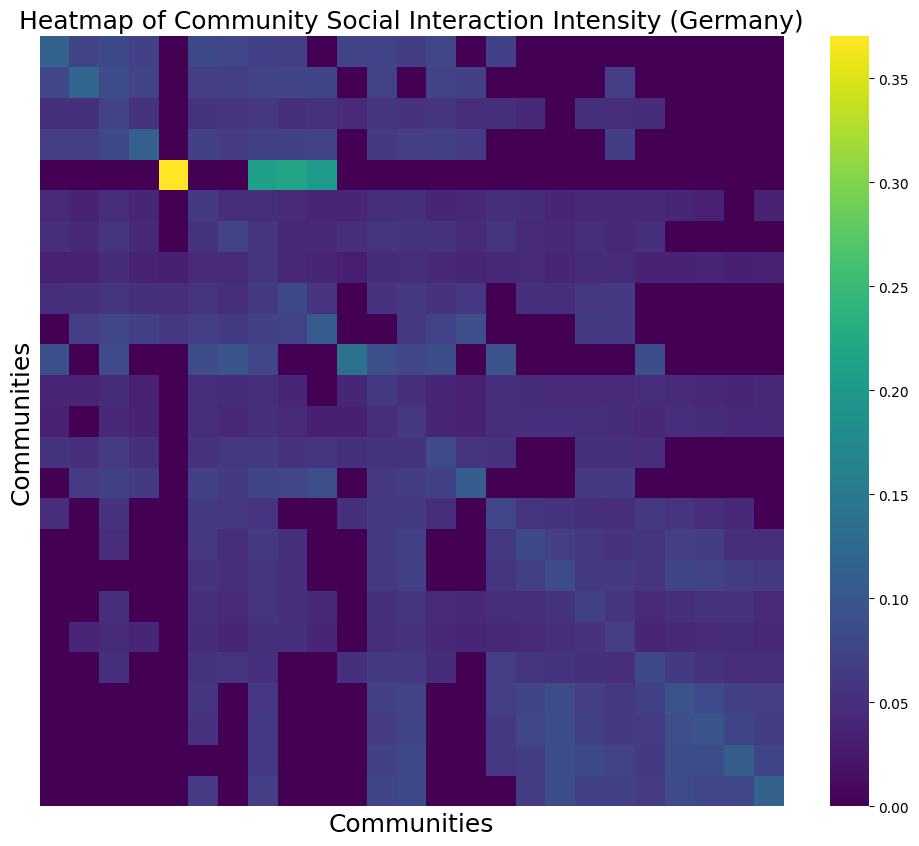}
	\caption{Heatmap of the social interaction matrix $W$ for Germany. Color intensity indicates the strength of opinion alignment between mobility-age communities, with lighter shades representing stronger social proximity.}
	\label{fig:heatmap}
\end{figure}

We now incorporate the insights obtained from the survey data into the model \eqref{eq:vector_model}. These insights include the community distribution $f_{m,p}$, the similarity matrix $W$ for the social network, and estimates of key parameters. This yields a data-driven version of the adoption-opinion model, that we can explore through numerical simulations to study its behavior under realistic conditions.

\section{Comparison between different targeted control strategies}\label{sec:comparison}
In this section, we use simulations to study how the coupled adoption-opinion system behaves under different scenarios. First, we examine the system's natural uncontrolled dynamics. Then, we evaluate the effects of various targeted interventions designed to promote the adoption of electric vehicles. We compare the interventions' relative efficacy and impact on different communities.

\subsection{Simulation Setting and Baseline Dynamics}
In the following analysis, we focus on the case of Germany, considering a total population of $N = 1078$ individuals, divided into $n = 25$ distinct communities. To facilitate interpretation, we report results as aggregate fractions across all communities, rather than community-specific values. The initial conditions for the simulations are derived from survey data. The initial fraction of adopters reflects the proportion of respondents who answered positively to the question: \emph{"Are there any electric cars or plug-in hybrid cars in your household?"}. This results in an initial adoption level of $19.4\%$ of the total population. The remainder is assumed to be susceptible, and no individuals are initially classified as dissatisfied. This setup represents a realistic scenario in which EV adoption is still limited and the population has had minimal exposure that could generate dissatisfaction.
We estimate the initial opinion distribution by determining the fraction of respondents in each community who agreed with the statement: \emph{"Society must reward electric cars instead of petrol and diesel cars"}. The average opinion across the network is approximately $0.537$. The dismissal parameter $\delta_{m,p}$ is also inferred from the survey data: for each community, it corresponds to the fraction of individuals who report having previously driven an EV but indicate they do not wish to own one in the present or future. This metric provides a community-level measure of disaffection. While, the adoption rate is set to $\beta = 0.01$ and the reversion rate to $\gamma = 0.02$. These values are chosen to produce a realistic progression of the uncontrolled system, capturing the natural evolution of adoption and the emergence of dissatisfaction over time without external intervention. For the opinion dynamics, the weights of the individual components are randomly selected for all communities, ensuring a balance between peer influence and behavioral feedback.

Figure~\ref{fig:nocontrol} shows the resulting dynamics in the absence of external intervention. In this uncontrolled scenario, the opinions vector converges to $x^*$ such that $R_{0}^A(x^*) = 1.32>1$. 
The system stabilizes at a positive equilibrium with a final fraction of adopters of approximately $0.06$, while about $0.077$ of the population becomes dissatisfied. The average final opinion at the equilibrium is around $0.37$, indicating a decreased value with respect to the initial condition. These results suggest that, without policy incentives or other external support, adoption remains very limited, and dissatisfaction naturally emerges over time. This baseline scenario underscores the challenge of promoting EV adoption and the potential effectiveness of targeted control strategies.\medskip

\begin{figure}
	\centering
	\includegraphics[width=0.35\textwidth]{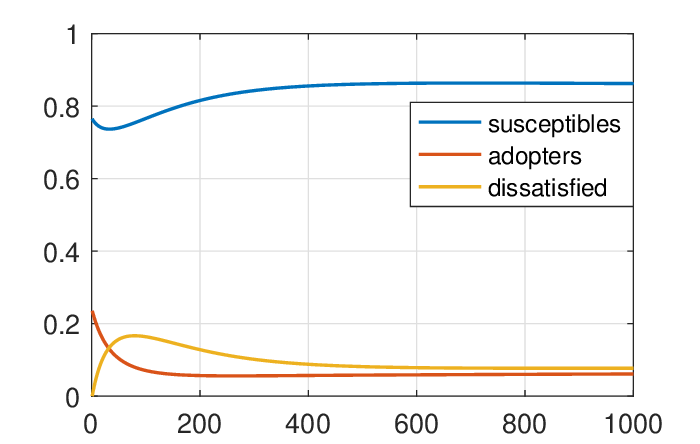}
	\caption{Numerical simulation of the aggregate uncontrolled dynamics \eqref{eq:vector_model} with initial condition derived from the survey data.}
	\label{fig:nocontrol}
\end{figure}

\begin{figure}
	\centering
	\hspace{-0.45cm}
	\subfloat[][]{\includegraphics[width=0.45\columnwidth]{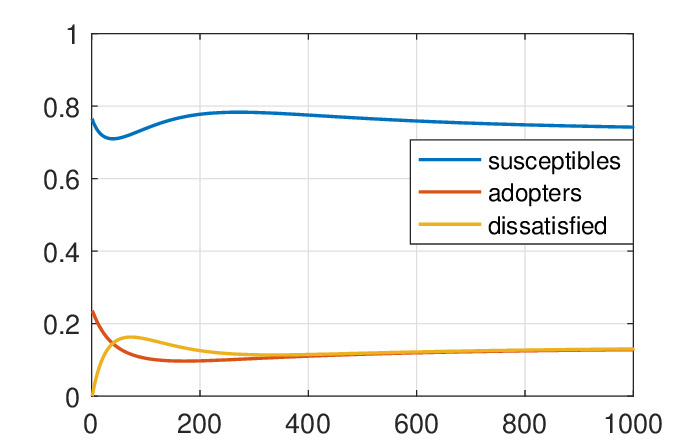}}
	\subfloat[][]{\includegraphics[width=0.45\columnwidth]{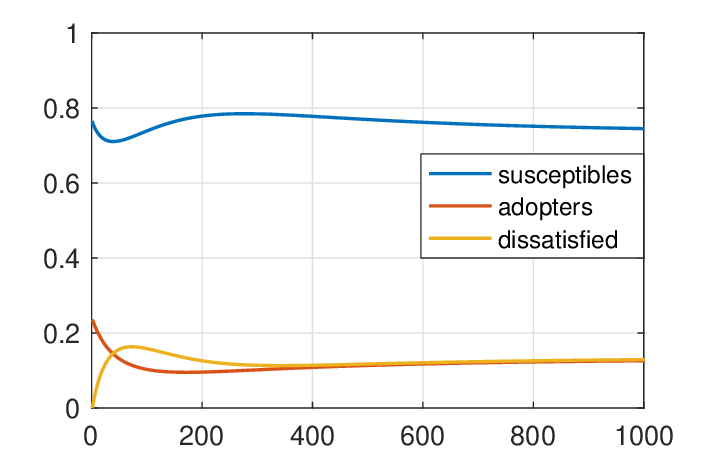}}\\
	\hspace{-0.45cm}
	\subfloat[][]{\includegraphics[width=0.45\columnwidth]{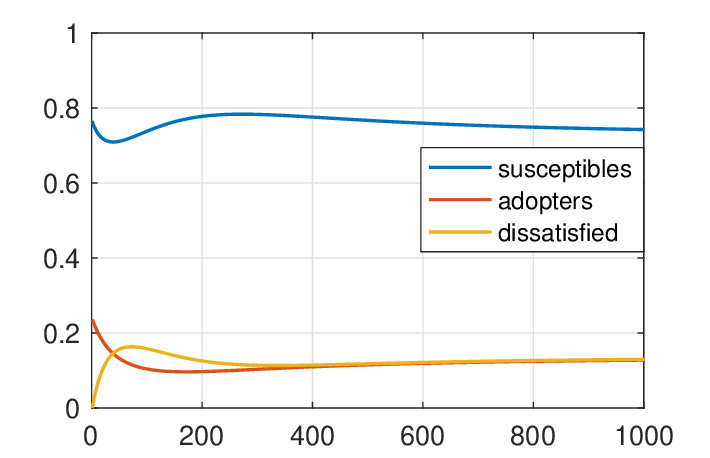}}
	\subfloat[][]{\includegraphics[width=0.45\columnwidth]{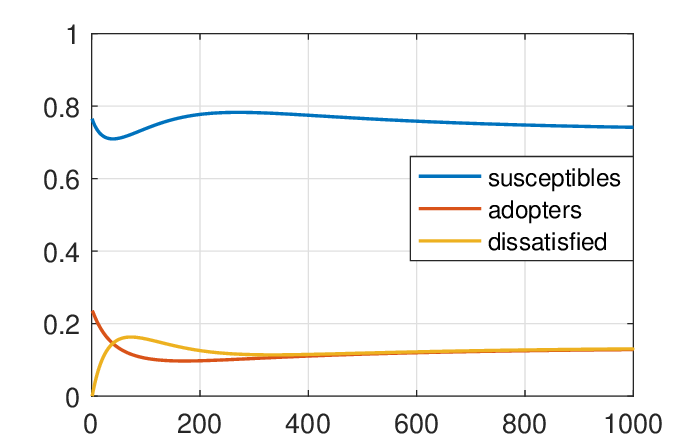}}
	\caption{Numerical simulation of the aggregate controlled dynamics with \eqref{eq:opinion-control}. In Plot (a), the budget allocation rule is dependent only on the relative size of each community, while in Plot (b), the distribution is weighted by the mobility index. In Plot (c) and (d) the budget allocation rule relies on centrality measures: in-degree and PageRank, respectively.}
	\label{fig:control2}
\end{figure}

\subsection{Opinion-based control policies}
One natural way to promote the adoption of electric vehicles is to influence opinion dynamics because social attitudes strongly affect individuals’ willingness to adopt new technologies. Targeted interventions, such as informational campaigns, educational programs, and public communications, can shift perceptions in favor of innovation, creating a positive feedback that amplifies adoption. Empirical studies confirm this effect: for example, in \cite{filippini2021nudging} the authors found that informational nudges on costs and environmental impacts significantly influence purchase intentions, while the work in \cite{li2020analysis} shows that popular science campaigns can increase the initial share of favorable attitudes in consumer networks. These results underscore the effectiveness of shaping beliefs and perceptions in steering collective behavior.
Motivated by these insights, we explore targeted control strategies designed to accelerate EV adoption by influencing the social layer. The idea is to modify the opinion update process to steer communities toward more favorable attitudes, thereby indirectly boosting adoption through the coupled adoption-opinion dynamics. This interplay between the opinion and behavioral layers is essential because shifting collective perceptions can trigger broader adoption cascades that propagate through social influence and mobility-driven interactions.
Formally, we implement opinion-based interventions by modifying the opinion update rule as follows:
\be \label{eq:opinion-control} x(t+1) = (I - \Lambda - \Xi)(x(0) + u) + \Lambda {W} x(t) + C \Xi \1 \mathrm{m}^T a(t),\ee
where $u$ represents the applied control vector. We assume a total available budget $U $ to be distributed across the population. 

A natural baseline allocation is proportional to community size, 
\be\label{eq:rule1}u_j = U f_j,\quad \forall j \in \mc V ,\ee
where $f_j$ is the relative size of community $j$. This reflects the idea of distributing resources evenly per capita across the entire population, without favoring specific groups. Another principle is to weight allocations by mobility, 
\be \label{eq:rule2}u_j = U\frac{ \mathrm{m}_j f_j}{\sum_k \mathrm{m}_k f_k}, \quad \forall j \in \mc V\ee
thus prioritizing communities with higher mobility. 

The rationale is that more mobile individuals are not only more visible in society but also more likely to shape social norms and accelerate the diffusion of new behaviors. 
Figures~\ref{fig:control2}(a)-(b) compare the resulting dynamics under the two allocation strategies, assuming a fixed budget equal to $n$. In both cases, the fraction of adopters stabilizes around $0.128$, with a similar final fraction of dissatisfied individuals. These results suggest that weighting the intervention by mobility does not significantly affect the adoption-dissatisfaction balance compared to the simpler size-proportional allocation. The similarity in outcomes can be attributed to the global nature of the network interactions, which dilute the impact of mobility-based reweighting, leaving the overall equilibrium essentially unchanged.
It is important to note that these comparisons are qualitative rather than quantitative: the goal is to highlight the relative tendencies of different allocation rules under a fixed budget. Overall, the analysis shows that opinion-based interventions can increase adoption relative to the uncontrolled baseline, although they may introduce a trade-off with higher dissatisfaction.

An alternative strategy for budget allocation relies on centrality measures derived from the similarity matrix $W$. In Figures~\ref{fig:control2}(c)-(d) illustrate two examples: in-degree centrality, which quantifies how much a given node is influenced by others \cite{Freeman1978CentralityIS}, and PageRank, a more refined metric that also considers the importance of the influencing neighbors \cite{Page1999ThePC}. Using either measure, the budget is distributed proportionally to the centrality values of the nodes. This approach yields a final fraction of adopters and dissatisfied individuals of approximately $0.129$ in the first case and $0.13$ in the second one, values slightly higher than in the previous cases, yet overall comparable in magnitude.
This outcome suggests that, in this network, in-degree and PageRank centrality identify nearly the same subset of structurally important nodes. In particular, nodes that are highly exposed to external influence also tend to attain high PageRank scores. Consequently, concentrating resources on these nodes produces analogous effects on both adoption and dissatisfaction levels.

To further examine the evolution of the population’s attitude, Figure~\ref{fig:control-op} reports the average opinion dynamics resulting from the different control strategies. As shown, the strategies illustrated in Figures~\ref{fig:control2}(a), (c), and (d) yield an average equilibrium opinion of approximately $0.68$, whereas the allocation rule that prioritizes communities with higher mobility leads to a slightly lower final value of about $0.66$. Overall, these results indicate a clear improvement in the collective attitude relative to the initial condition, confirming a positive effect of the implemented control interventions on opinion formation.
\begin{figure}
	\centering
	\includegraphics[width=0.3\textwidth]{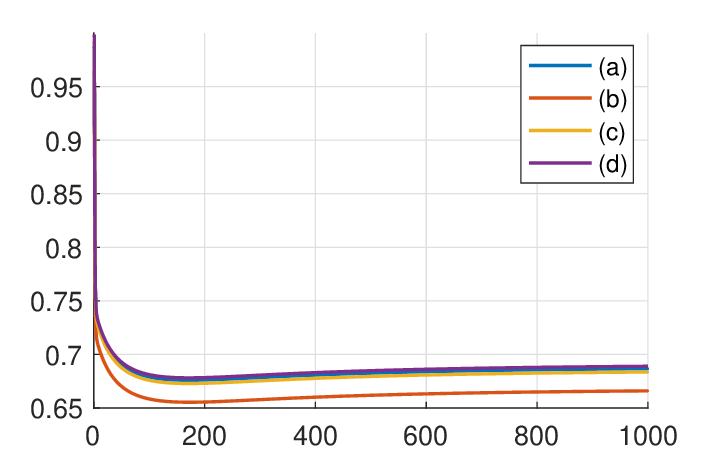}
	\caption{Numerical simulation of the aggregate controlled opinion dynamics \eqref{eq:opinion-control} under different allocation rules.}
	\label{fig:control-op}
\end{figure}

\subsection{Dissatisfaction-based policies}
Another way to influence adoption is to act directly on the experience of those who adopt. Minimizing dissatisfaction is key to sustaining adoption over time. If negative experiences accumulate, individuals may abandon the innovation, hindering its long-term diffusion \cite{rezvani2015advances}. Improving product quality, providing reliable customer service, and enhancing user experience are all measures that reduce dissatisfaction. 
Inspired by these findings, we model this control strategy through a modified dissatisfaction rate:
\be \label{eq:delta-control}\delta^u_i := \delta_i (1 - u_i(t)).\ee

We performed a similar analysis to the one presented earlier for the opinion-based control, but under a smaller total budget fixed at $0.75n$. The reason is that interventions aimed at improving the user experience, such as better interfaces, increased reliability, and post-adoption support, are usually more expensive per individual than information campaigns. Thus, a smaller budget is more realistic in this context.
Figures~\ref{fig:control4}(a)-(b) illustrate the controlled scenarios under two baseline allocation rules: proportional to community size \eqref{eq:rule1} and weighted by mobility \eqref{eq:rule2}. 
In the first case, the system converges to a final adoption fraction of approximately $0.66$, with dissatisfaction stabilizing around $ 0.15$. When mobility is considered, adoption further increases to $0.69$, with dissatisfaction slightly decreases to $0.126$. From a high-level perspective, both strategies outperform opinion-based interventions, suggesting that interventions directly addressing dissatisfaction produce a stronger and more sustainable impact on overall adoption levels.
Figures~\ref{fig:control4}(c)-(d) show the results when the budget is instead allocated according to centrality measures derived from the similarity matrix $W$. Targeting nodes with higher in-degree centrality yields a final adoption fraction of $0.78$ and a dissatisfaction fraction of $0.1$. When PageRank centrality is employed, these values further improve to $0.8$ and $0.1$, respectively. These results highlight a significant improvement compared to both the uncontrolled baseline and the earlier opinion-based interventions.

\begin{figure}
	\centering
	\hspace{-0.45cm}
	\subfloat[][]{\includegraphics[width=0.45\columnwidth]{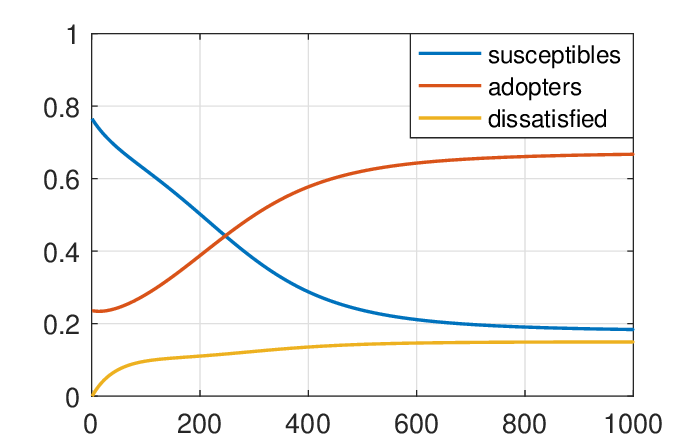}}
	\hspace{-0.35cm}
	\subfloat[][]{\includegraphics[width=0.45\columnwidth]{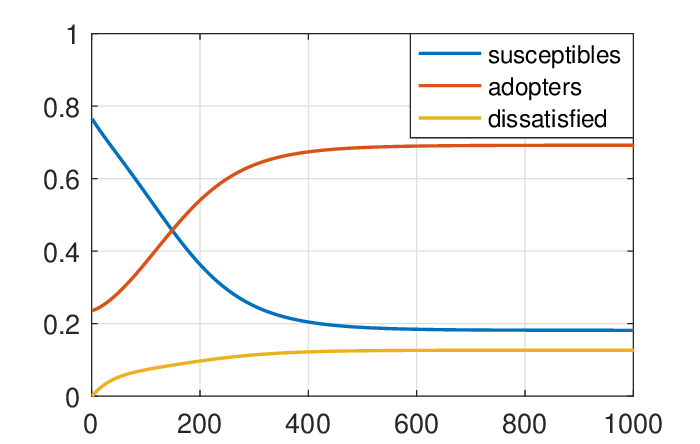}}\\
	\hspace{-0.45cm}
	\subfloat[][]{\includegraphics[width=0.45\columnwidth]{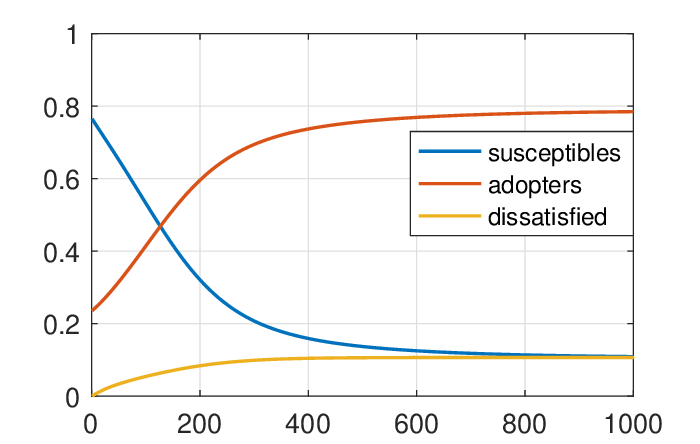}}
	\hspace{-0.35cm}
	\subfloat[][]{\includegraphics[width=0.45\columnwidth]{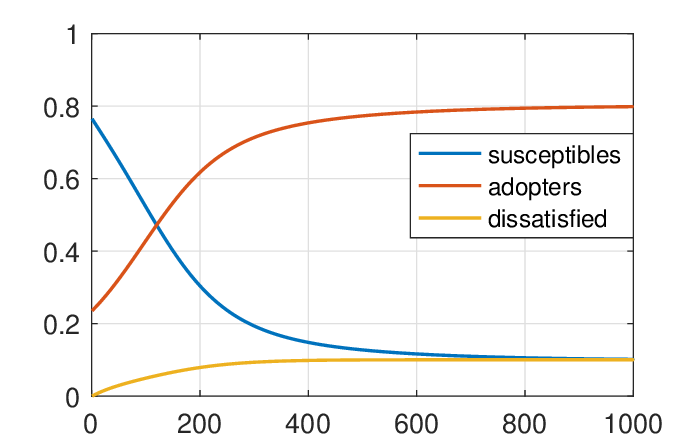}}
	\caption{Numerical simulation of the aggregate controlled dynamics with \eqref{eq:delta-control} as modified dissatisfaction rate. In Plot (a), the budget allocation rule is dependent only on the relative size of each community, while in Plot (b), the distribution is weighted by the mobility index. In Plot (c) and (d) the budget allocation rule relies on centrality measures: in-degree and PageRank, respectively.}
	\label{fig:control4}
\end{figure}

As in the previous analysis, we also examine how these control strategies affect the evolution of opinions. Figure~\ref{fig:control-op2} displays the average opinion dynamics corresponding to the different allocation rules. In this case, the control strategies shown in Figures~\ref{fig:control4}(b), (c), and (d) yield a higher average equilibrium opinion compared to the allocation rule proportional solely to community size \eqref{eq:rule1}. In all scenarios, the interventions lead to an overall improvement in the population’s attitude with respect to the initial condition.
\begin{figure}
	\centering
	\includegraphics[width=0.3\textwidth]{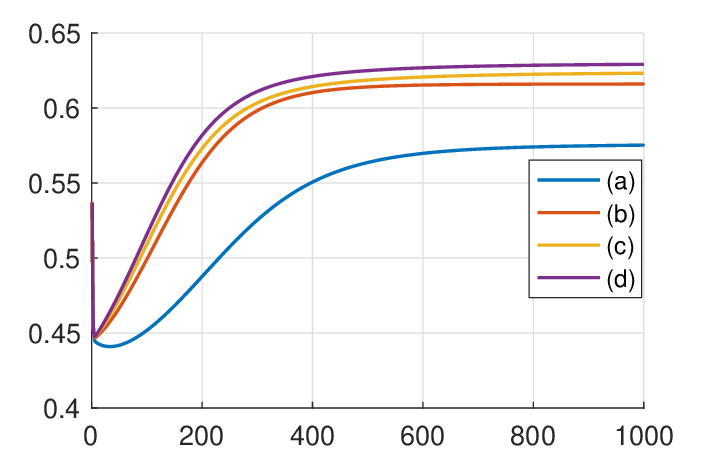}
	\caption{Numerical simulation of the aggregate controlled opinion dynamics with \eqref{eq:delta-control} as modified dissatisfaction rate under different allocation rules.}
	\label{fig:control-op2}
\end{figure}

Overall, these findings suggest that improving the quality of the adoption experience is more effective than merely shaping opinions. Although opinion-based interventions can trigger adoption, they often lead to increased dissatisfaction. In contrast, dissatisfaction-based interventions lead to higher final adoption levels and sustain diffusion more robustly over time. Furthermore, in the second control scenario, we observe that targeting based on mobility improves efficiency and that leveraging network centrality amplifies this effect, ensuring resources focus on the most influential communities. This layered perspective suggests that combining improvements in user experience with informed allocation rules is the most promising way to maximize long-term diffusion while keeping low dissatisfaction levels.\medskip

\section{Conclusion}\label{sec:conclusion}
In this paper, we proposed a data-driven computational framework for studying the diffusion of sustainable behaviors, with a focus on electric vehicle adoption. We built a synthetic population and a social network capturing inter-community influence by integrating survey-based opinion data with mobility and socio-demographic features. We then analyzed a coupled adoption-opinion model to explore both spontaneous dynamics and the impact of different interventions.
Our results show that individual decisions are only slightly influenced by promotional or awareness campaigns. Rather, adoption dynamics are primarily shaped by personal experience and perceived satisfaction. Indeed, although policies that act only on opinion dynamics can accelerate diffusion, they also tend to exacerbate dissatisfaction. In contrast, measures that improve user experience generate positive feedback that sustains adoption over time. Overall, these findings suggest that transitioning to sustainable technologies requires more than persuasive communication; it demands concrete improvements in user experience and systemic support to reinforce satisfaction and trust.

Future work could extend this framework by incorporating richer behavioral feedback mechanisms, exploring dynamic policy adaptation, and integrating additional data sources to refine the social network structure. These extensions would further enhance the model’s predictive power and its utility for designing effective, evidence-based policies to promote sustainable innovation.

\bibliographystyle{IEEEtran}
\bibliography{bib}

\appendix

\begin{proof}[Proof of Theorem \ref{theo:stability}]
	(i) The proof can be partially adapted from \cite[Theorem 1]{alutto2025predictivecontrolstrategiessustaining}. With respect to the model in \cite{alutto2025predictivecontrolstrategiessustaining}, here we have $\theta_j = 0$, thus this implies the adoption-free equilibrium reduces to $(\0, \0, x^*)$ with $x^* = (I - \Lambda W)^{-1} (I - \Lambda - \Xi) x(0)$.
	
	Points (a) and (b) can be deduced by the linearization analysis in this simplified setting. 
	
	Consider the matrix $M(x) = I - \Delta +c\beta \mathrm{diag}(x) \mathrm{m}\, \mathrm{m}^T$, which is irreducible, nonnegative and monotonically increasing with respect to $x$. The adopters' dynamics can then be upper-bounded as
	$ \label{eq:eq1}a(t+1) \leq M(\bar x) a(t)\,.$
	Since the maximal reproduction number satisfies $R_{0,\max}^A = M(\bar x) <1$, then by the Perron-Frobenius theorem, there exists a nonnegative vector $v\neq0$ and a constant $\lambda\in(0,1)$ such that $v^TM(\bar x)= \lambda v^T$. Multiplying both sides of \eqref{eq:eq1} by $v^T$ gives
	$$ v^Ta(t+1) \leq v^T M(\bar x) a(t) \leq  \lambda v^T a(t),$$
	which implies $a(t)$ asymptotically converges to $\0$ for any initial state $a(0)$ in $[0,1]^{\V}$.	Note that from Assumption \ref{ass:ass1} it follows that $\Lambda W$ is a Schur stable matrix (see Lemma 5 in \cite{FRASCA2013212}) and the opinions vector $x$ asymptotically converges to $ (I- \Lambda W)^{-1} (I-\Lambda-\Xi) x(0)$, since it is the unique asymptotically stable equilibrium of the resulting Friedkin-Johnsen model.
	It remains to show that $d(t)$ asymptotically converges to $\0$. 
	Defining the error
	$$e_i(t) =d_i(t) - \frac{\delta_i a_i(t)}{ \gamma x_i(t)}\,,$$
	for all $i \in \V$, we get
	$ e (t+ 1) =U(t)e (t) + b(t)\,$
	with $U(t)= \big( I- \gamma \mathrm{diag}(x(t))\big)$ and 
	\begin{align*}
		b_i (t)& =\! \left[ \frac{\delta_i}{\gamma x_i(t)} \!+\! \frac{-\! \delta_i (1 \!-\! \delta_i)}{\gamma x_i(t+ 1)} \right]\!\! a_i(t) \\[3pt]
		&+ \frac{-\! \delta_i (1 \!-\! a_i(t) - d_i(t))}{\gamma x_i(t+ 1)} c \beta m_i x_i(t)\!\! \sum_{j \in \mc V} \! m_j a_j(t) \,.
	\end{align*}
	We have $\|U(t)\|_\infty \leq 1-\min_i(\gamma x_i) <1$. Since $a(t)$ asymptotically converges to $0$ and $x(t)$ is also convergent to $x^*$, it follows that $\|b(t)\|_\infty$ vanishes as $t \to \infty$. Thus proving Point (c).\\

	(ii) The proof of the existence of an equilibrium point $(a^\dag, d^\dag, x^\dag)$ with $a^\dag>\0$ can be adapted from \cite[Theorem 2]{alutto2025predictivecontrolstrategiessustaining}. 
	
	Let us prove now the stability of  such equilibrium point. Define the error variables $e^a_i(t) = a_i(t)- a_i^\dag$ and $e_i^d(t) = d_i(t) - d_i^\dag$. Then, 
	{\small{\begin{align*}
				\mspace{-5mu}e_i^a(t+1)\! &\!= \!a_i(t) + c \beta x_i(t) m_i (1\!-\! a_i(t)\!-\! d_i(t))\!\! \sum_{j \in \mc V}\! m_j a_j(t) +\\ 
				&\mspace{12mu}- \delta_i a_i(t) - a_j^\dag \\
				& \!=\! e_i^a(t) + c\beta x_i(t) m_i (1 -a_i^\dag -d_i^\dag) \!\! \sum_{j \in \mc V} \! m_j (e_j^a(t) + a_j^\dag) +\\ 
				& \mspace{12mu} - \! c \beta x_i(t) m_i \! (e_i^a(t)\!+\!e_i^d(t))\!\! \sum_{j \in \mc V}\! m_j \! (e_j^a(t) \!+ \! a_j^\dag) \!-\! \delta_i (e_i^a(t) \!+\! a_i^\dag) \\
				&\!= \Big(1-\delta_i-  c \beta x_i(t) m_i \sum_{j \in \mc V} m_j a_j^\dag\Big)e_i^a(t) +\\
				&\mspace{12mu}+c \beta x_i(t)m_i (1 - a_i(t) - d_i(t)) \sum_{j \in \mc V}m_j e_j^a(t) +\\
				&\mspace{12mu}-c \beta x_i(t) m_i e_i^d(t) \sum_{j \in \mc V} m_j a_j^\dag .
	\end{align*}}}
	While
	{\small	\begin{align*}
			e_i^d(t+1) &= d_i(t+1) - d_i^\dag \\
			&= e_i^d(t) \!+\! \delta_i a_i(t) \!-\! \gamma_i x_i(t) d_i(t)  \\
			&= e_i^d(t) + \delta_i (e_i^a(t) + a_i^\dag) - \gamma_i x_i(t) (e_i^d(t) + d_i^\dag) \\
			&= (1 - \gamma_i x_i(t)) e_i^d(t) + \delta_i e_i^a(t).
	\end{align*}}
	Therefore, we can rewrite the systems for $e^a$ and $e^d$ in the following compact form 
	\begin{equation}\label{eq:system2}
		\begin{bmatrix}
			e^a(t+1) \\
			e^d(t+1)
		\end{bmatrix} = \begin{bmatrix}
			F_{11}(t) & F_{12}(t) \\
			F_{21}(t) & F_{22}(t)
		\end{bmatrix}\begin{bmatrix}
			e^a(t) \\
			e^d(t)
		\end{bmatrix} 
	\end{equation}
	where
	\begin{align*}
		F_{11}(t) &=  I- \Delta - \mathcal{B}^\dag +  c \beta M \mathrm{diag}(x) \mathrm{diag}(\1 \!-\! a\! -\! d)\! \1 \! \1^T \!  M, \\
		F_{12}(t) &= -\mc B^\dag, \\
		F_{21}(t) &=  \Delta, \\
		F_{22}(t) &=I - \Gamma \mathrm{diag}(x(t)).
	\end{align*}
	
	In order to study the stability of system~\eqref{eq:system2}, we first analyze the infinity norms of the block matrices appearing in its dynamics, and we show that they are uniformly bounded in time.
	In particular, by the definition in \eqref{eq:varphi}, and since $x(t),a(t),d(t)\in [0,1]^{\V}$, we have that for all $t \geq 0$, $\varphi=  \sup_t \|F_{11}(t)\|_{\infty} < 2$. 
	Similarly, we have $\eta = \sup_t \|F_{22}(t)\|_\infty < 1$ and $\nu= \sup_t \|F_{21}(t)\|_\infty < 1$. Moreover, let $b := \sup_t \|F_{12}(t)\|_\infty$, and note that, from Proposition \ref{prop:invariant}, we have  $b<1$.\\
	From the system \eqref{eq:system2}, we can write
	\begin{align}
		\| e^a(t+1)\|_\infty &= \| F_{11}(t) e^a(t) + F_{12}(t) e^d(t)\|_\infty \nonumber\\
		&\leq \varphi \|e^a(t)\|_\infty + b\|e^d(t)\|_\infty, \label{eq:bound1}
	\end{align}
	and 
	\begin{align}
		\| e^d(t+1)\|_\infty &= \| F_{21}(t) e^a(t) + F_{22}(t) e^d(t)\|_\infty\nonumber \\
		&\leq \nu \|e^a(t)\|_\infty + \eta \|e^d(t)\|_\infty. \label{eq:bound2}
	\end{align}
	Let us now define the auxiliary variable 
	$$g(t):= \begin{bmatrix} \|e^a(t)\|_\infty \quad \sqrt{ \tfrac{b}{\nu}} \|e^d(t)\|_\infty \end{bmatrix}^\top.$$
	Considering the bounds in \eqref{eq:bound1} and \eqref{eq:bound2}, we get the following auxiliary system:
	\begin{equation}\label{eq:aux-system}g(t+1) \leq \,G \, g(t),\end{equation}
	where $G$ is defined in \eqref{eq:G}. Note that matrix $G$ is nonnegative and symmetric, hence, by the Perron-Frobenius theorem, its spectral radius $\rho(G)$ coincides with its induced $\infty$-norm. Therefore,
	$$\| g(t+1)\|_\infty \leq \|G\|_\infty \| g(t)\|_\infty  = \rho(G)  \| g(t)\|_\infty.$$
	Since $\rho(G)<1$, then $g(t)$ goes to $ 0$ exponentially as $t$ tends to $\infty$. Consequently, both $\|e^a(t)\|_\infty$ and $\|e^d(t)\|_\infty$ converge exponentially to zero, and thus the system \eqref{eq:system2} is exponentially stable.
	
	Note again that from Assumption \ref{ass:ass1} it follows that $\Lambda\tilde{W}$ is a Schur stable matrix (see Lemma 5 in \cite{FRASCA2013212}) and the opinions vector $x$ asymptotically converges to $x^\dag = (I- \Lambda W)^{-1} \big((I-\Lambda-\Xi) x(0) + c \Xi \1 \mathrm{m}^T a^\dag)$, since it is the unique asymptotically stable equilibrium of the resulting Friedkin-Johnsen model.
\end{proof}

\begin{IEEEbiography}[{\includegraphics[width=1in,height=1.25in,clip,keepaspectratio]{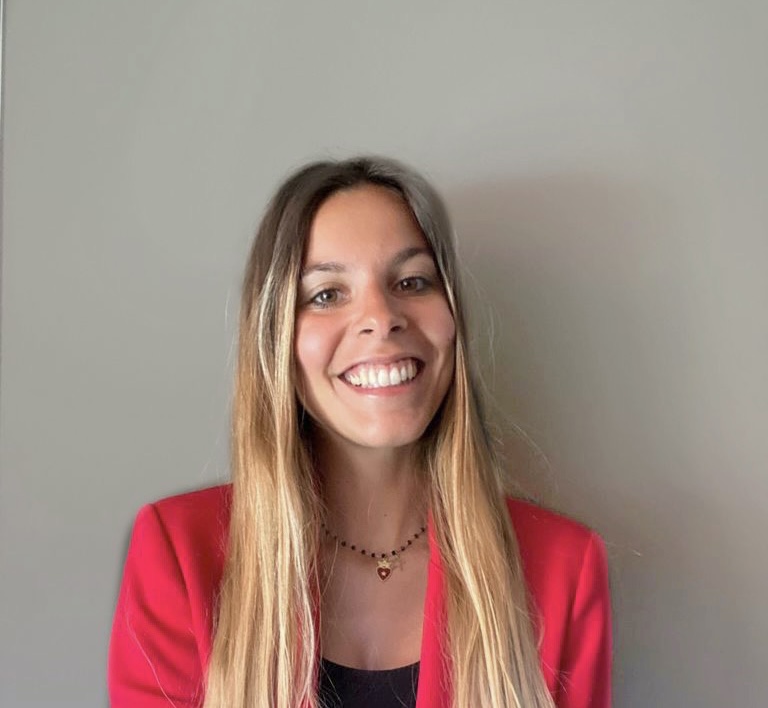}}]{Martina Alutto} received the B.Sc.~and the M.S.~(cum laude) in Mathematical Engineering from  Politecnico di Torino, Italy, in  2018 and 2021, respectively and the PhD~(cum laude) in Pure and Applied Mathematics at the Department of Mathematical Sciences, Politecnico di Torino, Italy. She was a Research Assistant at the National Research Council (CNR-IEIIT), Torino, Italy. She is currently a Postdoctoral Researcher with the Royal Institute of Technology, Stockholm, Sweden. She was Visiting Student at Cornell University, Ithaca, NY in 2023. Her research interests focus on analysis and control of network systems, with application to epidemics and social networks.
\end{IEEEbiography}
\begin{IEEEbiography}[{\includegraphics[width=1in,height=1.25in,clip,keepaspectratio]{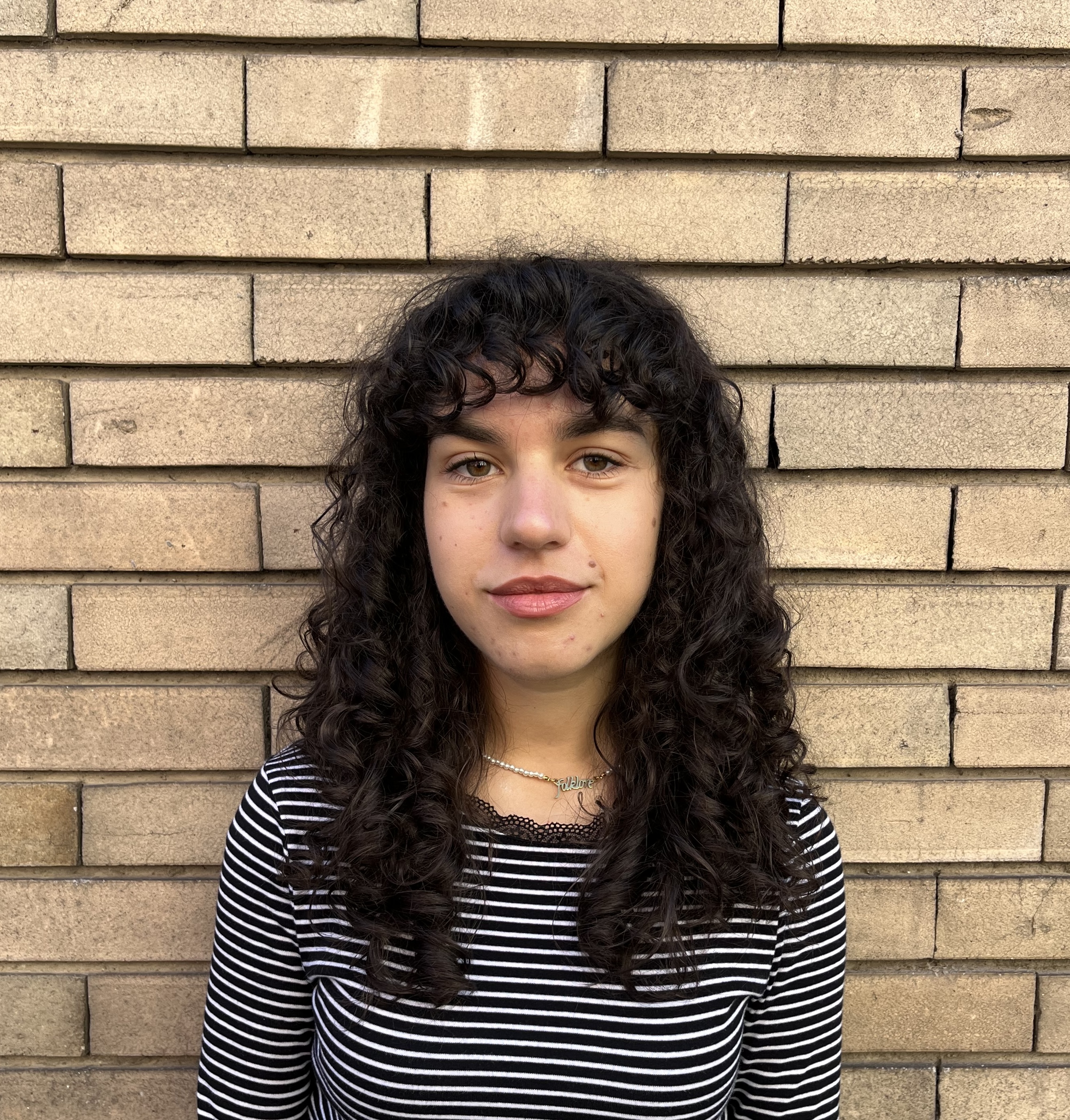}}]{Sofia Bellotti} received the B.Sc. degree in Mathematics for Engineering from Politecnico di Torino, Turin, Italy, in 2023. She spent a semester at Centrale Supélec, Paris, France, one of the leading Grandes Écoles d’Ingénieurs. She is currently pursuing the M.Sc. degree in Mathematical Engineering at Politecnico di Torino and preparing her master’s thesis at ENS Paris-Saclay, Paris, France, on conformance checking for stochastic and time-aware process models. Her research interests include data analysis, optimization, and machine learning.
\end{IEEEbiography}
\begin{IEEEbiography}[{\includegraphics[width=1in,height=1.25in,clip,keepaspectratio]{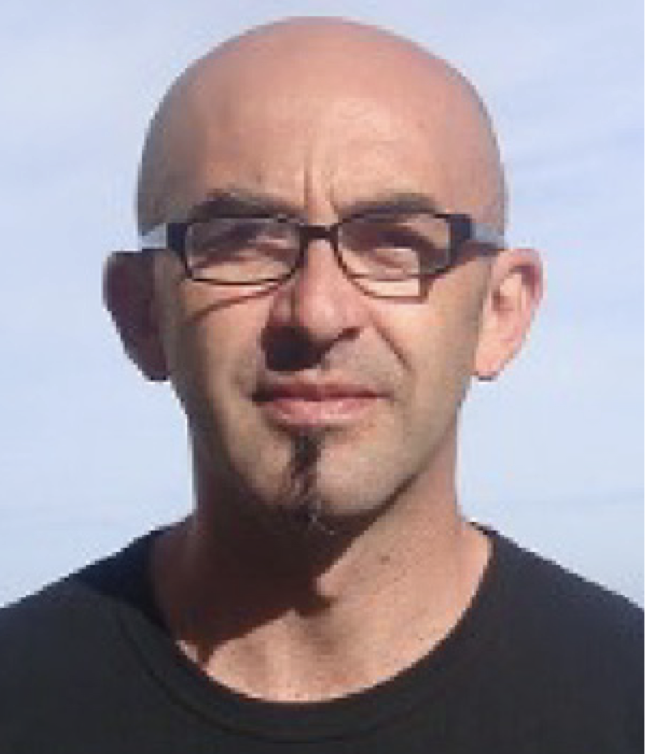}}]{Fabrizio Dabbene} is a Director of Research at the institute IEIIT of the National Research Council of Italy (CNR), where he coordinates the Information and Systems Engineering Group. He has held visiting and research positions with The University of Iowa, Penn State University, and the Russian Academy of Sciences, Institute of Control Science, Moscow, Russia. He has authored or coauthored more than 150 research papers and two books. Dr. Dabbene was an Elected Member of the Board of Governors, from 2014 to 2016. He has served as the vice president for publications, from 2015 to 2016. He has also served as an Associate Editor for Automatica (2008–2014), and IEEE Transactions on Automatic Control (2008-2012) and as Senior Editor of the IEEE Control Systems Society Letters (2018–2023), and he is currently Senior Editor for the IEEE Transactions on Control Systems Technology. He chaired the IEEE-CSS Italy Chapter (2019–2024) and since 2023 he serves as NMO representative for Italy at the International Federation of Automatic Control (IFAC). He is recipient of the 2024 IEEE CSS Distinguished Member Award.
\end{IEEEbiography}
\begin{IEEEbiography}[{\includegraphics[width=1in,height=1.25in,clip,keepaspectratio]{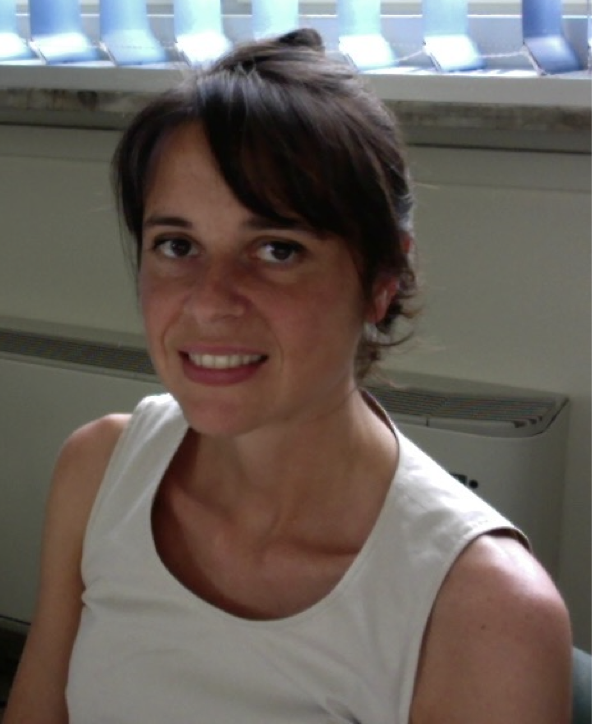}}]{Chiara Ravazzi} is a Senior Researcher at the Italian National Research Council (CNR-IEIIT) and adjunct professor at the Politecnico di Torino. She obtained the Ph.D. in Mathematical Engineering from Politecnico di Torino in 2011. In 2010, she spent a semester as a visiting scholar at the Massachusetts Institute of Technology (LIDS), and from 2011 to 2016, she worked as a post-doctoral researcher at Politecnico di Torino (DISMA, DET). She joined the Institute of Electronics and Information Engineering and Telecommunications (IEIIT) of the National Research Council (CNR) in the role of a Tenured Researcher (2017-2022). Furthermore, she served as an Associate Editor for IEEE Transactions on Signal Processing from 2019 to 2023, and she currently holds the same position for IEEE Transactions on Control Systems Letters (since 2021) and the European Journal of Control (since 2023). She has achieved the national scientific qualification as a full professor in the field of Automatica (09/G1) and as Associate Professor in the area of Telecommunications (09/F2).
\end{IEEEbiography}
\end{document}